\documentclass[notitlepage, final, 10pt, pra, aps, twocolumn, floatfix, showpacs, letter]{revtex4-1}
\usepackage{graphicx}
\usepackage{adjustbox}
\usepackage{dcolumn}
\usepackage{bm}
\usepackage{dsfont}
\usepackage[colorlinks,citecolor=blue,linkcolor=black,urlcolor=black,hyperindex]{hyperref}
\usepackage{dsfont}
\usepackage{physics}
\usepackage{array, booktabs}
\usepackage{braket}
\usepackage{amsmath, amssymb, amsthm}
\usepackage{mathrsfs}
\usepackage[protrusion=true, expansion=true]{microtype}
\newtheorem{theorem}{Theorem}
\newtheorem{corollary}{Corollary}[theorem]
\newtheorem{lemma}[theorem]{Lemma}
\newtheorem{definition}{Definition}

\newcommand{\mc}{\mathcal}

\newcolumntype{M}[1]{>{\centering\arraybackslash}m{#1}}

\DeclareMathOperator*{\E}{\mathbb{E}}
\DeclareMathOperator*{\EF}{\mathbb{E}[\textit{F} \hspace{2pt}]}

\begin{document}

\definecolor{purple}{rgb}{0.7,0,0.7}
\definecolor{dkgreen}{rgb}{0,0.6,0}
\definecolor{brown}{rgb}{0.8,0.4,0}
\definecolor{midnightblue}{rgb}{0.39,0.58,0.93}

\preprint{APS/123-QED}

\title{Analysis and Suppression of Errors in Quantum Random Access Memory under Extended Noise Models}

\author{Rohan Mehta}
\email{rohanmehta@uchicago.edu}
\affiliation{%
Pritzker School of Molecular Engineering, University of Chicago, Chicago, IL 60637, USA
}%
\author{Gideon Lee}%
\affiliation{%
Pritzker School of Molecular Engineering, University of Chicago, Chicago, IL 60637, USA
}%
\author{Liang Jiang}
\affiliation{%
Pritzker School of Molecular Engineering, University of Chicago, Chicago, IL 60637, USA
}%

\begin{abstract}
    Quantum random access memory (QRAM) is required  for numerous quantum algorithms and network architectures. Previous work has shown that the ubiquitous bucket-brigade QRAM is highly resilient to arbitrary local incoherent noise channels occurring during the operation of the QRAM [PRX Quantum 2, 020311 (2021)], with query infidelities growing only polylogarithmically with memory width when errors are assumed to only occur on individual routers. We extend this result to a large class of generalized settings that arise in realistic situations, including arbitrary initialization errors, spatially correlated errors, as well as coherent errors, maintaining the polylogarithmic scaling in all instances. Fully quantifying the extent to which QRAM’s noise resilience holds may provide a guide for the design of QRAM architectures — for instance, the resilience to initialization errors indicates that a reset protocol between successive queries may not be necessary. In the case of coherent errors, we find an up-to-quadratic increase in the infidelity bound, and therefore discuss generalizations to randomized compiling schemes, which usually are rendered inapplicable in the QRAM setting, to tailor these errors into more favorable stochastic noise.
\end{abstract}

\date{\today}
\maketitle

\section{\label{sec:level1}Introduction}
The concerted global effort towards the realization of quantum computers predicates on a quantum advantage that can be realized in comparison to classical devices. Many of the most interesting proposals for quantum algorithms predicate on solving longstanding, fundamentally classical problems: we provide the algorithm a description of a classical problem, and with high probability, we are returned a useful classical answer. Quantum processing can provide utility at intermediary stages to reduce resource consumption, but given the interest in these so-called classical end-to-end applications \cite{Tang2021}, interfaces between classical data and quantum processors are imperative. The primary module to achieve this, quantum random access memory (QRAM), an analogue of classical random access memory, has been presumed as an oracle to facilitate claims of quantum speedup in various algorithms \cite{Harrow2009, Schutzhold2003, Gui2024}. Moreover, QRAM-based schemes are prevalent in complexity-optimal quantum state preparation \cite{Casares2020, Ashhab2022}, quantum data centers \cite{Liu2023}, and resource-efficient metrology \cite{Khabiboulline2019, Khabiboulline2019_2}. Despite the importance of realizing a quantum random access memory, concrete constructions have been explored only recently, and an experimental demonstration of QRAM has yet to be established.

Classical RAM retrieves an entry $x_i$ from a database $\vec{x}$ of size $N = 2^n$ corresponding to an $n$ bit binary address state $i$. Conversely, a QRAM coherently queries a superposition of classical data based on a quantum address. As an oracle, a query to the QRAM acts on a quantum address concatenated with a $\ket{+}$ state bus qubit,
$\ket{\psi_{\textit{in}}} = \sum_{i = 0}^{N - 1} \alpha_i \ket{i}\ket{+}$, where each $i \in \mathbb{F}_2^n$ is the binary index corresponding to a particular entry in $\vec{x}$.
\begin{equation}
    O: \sum_{i = 0}^{N - 1} \alpha_i \ket{i}^{\mathcal{A}} \ket{+}^\mathcal{B} \mapsto  \sum_{i = 0}^{N - 1} \alpha_i \ket{i}^{\mathcal{A}} Z^{x_i}\ket{+}^\mathcal{B} \label{eq: QRAM Oracle}
\end{equation}
where, for simplicity, each $x_i$ is a classical bit of information, hence the need for only one bus qubit. We will define 
\begin{equation}
    \ket{\psi_{\textit{out}}} := O \ket{\psi_{\textit{in}}}
\end{equation}
Here, $\mathcal{A}, \mathcal{B}$ denote the address and bus registers, respectively. Though we will not dwell on these applications, modifications to aspects of the prototypical QRAM scheme can be made to accommodate quantum data, which of course cannot be copied but instead transferred due to no-cloning limitations. 

In order to achieve the claimed advantages of algorithms requiring a QRAM, it is crucial that the QRAM operation can be carried out with high fidelity. Yet, performing full error correction on QRAM is expected to extremely challenging (see for instance \cite{Hann2021, Matteo2020, Lee2023, Jaq2023} for a review of the outstanding challenge of error-corrected QRAM). Native noise suppression and mitigation is thus critical, as this will reduce the level of error correction necessary for QRAM to reach certain target error rates, ensuring that the overhead for fault-tolerant QRAM does not dominate the cost of scalable quantum computing. Reducing the base error rates of QRAM may also allow us to explore the applications of QRAM in the early fault-tolerance regime, where full error correction may not be available. One prong of this research program is to explore noise resilient experimental designs of the various components that make up a QRAM \cite{Hann_PRL_2019_hybrid_QRAM, Hong_PRA_2012_robust_QRAM, Moiseev_JMO_2016_timebin_QRAM, Wang_arxiv_2024_transmon_QRAM, Wang_PRApp_2021_SC_QRAM_expt, Chen_PRXQ_2021_spinphoton_QRAM, Weiss_PRXQ_2024_3D_QRAM, Asaka_QST_2021_walker_QRAM}. At the same time, from an architectural point of view, it is equally important to identify schemes that ensure that QRAM is highly noise-resilient even prior to error correction \cite{Giovannetti2008_Architectures, Giovannetti2008, Weiss_arxiv_2024_faulty_towers, XuTBA, Lee2023}. To this end, the most promising approach to building QRAM involves the bucket-brigade architecture \cite{Giovannetti2008_Architectures, Giovannetti2008} and variations thereof \cite{Xu2023, XuTBA}. In Section~\ref{sec: prelims}, we review prior work on the bucket-brigade architecture and its noise resilience, which constitute the primary object of study for this work. For now, it suffices to say that despite requiring a hardware overhead that goes as $O(N)$, bucket-brigade QRAM only suffers a query infidelity that goes as $O(\text{poly}(n))$, and can thus be said to be naturally robust \cite{Hann2021}.  Given how much hinges on the bucket-brigade's natural noise resilience, it is important to fully demarcate the bounds of this noise resilience. We will implicitly assume the bucket-brigade architecture whenever mentioning QRAM.

The contributions of this work to the research program of building robust QRAM falls into two main parts. The first part pertains to expanding known results on the natural noise resilience enjoyed by the bucket-brigade QRAM. Extending beyond the standard paradigm of local, uncorrelated, stochastic noise occurring during the operation of the QRAM, in Section~\ref{sec:extended_noise_resilience}, we greatly extend the scope of this result to show that QRAM also exhibits a natural noise resilience against a much larger class of errors that may arise in practice. This includes: (1) Spatially correlated noise across multiple routers (Section~\ref{sec: Correlated Errors}), (2) Coherent noise, in which the relative phase of Kraus operators on different routers must be taken into account (Section~\ref{sec: coherent noise}), and (3) Propagated errors between queries that arise from failing to reset the QRAM tree, for instance, due to hardware constraints (Section~\ref{sec: Arbitrary initializations}). In all the above cases, we show that QRAM retains its polylogarithmic error scaling, but the strength of the noise suppression depends on the class of the error. 

Beyond relying on the bucket-brigade's natural noise resilience, it is also essential to explore alternatives to full QEC that can further push down these error rates. Standard schemes for error mitigation \cite{Cai_RMP_2023_QEM_review} may be applied \cite{Shi_arxiv_2024_error_mitigated_qram}, but do not necessarily serve QRAM well; such schemes primarily produce expectation values , whereas the goal of QRAM is to produce $| \psi_{out} \rangle$ as a resource for computation. The complexity and size of the bucket-brigade architecture is such one must then look to bespoke means of suppressing errors in QRAM \cite{Weiss_arxiv_2024_faulty_towers, Lee2023, Weiss_PRXQ_2024_3D_QRAM}. Following this line of research, the second part of our work involves a QRAM-specific proposal for suppressing errors beyond the natural noise resilience that QRAM already enjoys. In Section~\ref{sec: Pauli twirling}, we consider Pauli twirling \cite{Wal2016}, which is frequently assumed in most circuit models of quantum computing in order to turn coherent errors into incoherent ones, which then scale more favourably and are more amenable to error correction. According to Section~\ref{sec: coherent noise}, this difference in the effect of coherent and incoherent errors also holds for QRAM.  However, due to the non-Clifford and entangling nature of QRAM, standard techniques for randomized compiling, which rely on the closure of the Pauli group under Clifford operations or single-qubit non-Cliffords, do not work. While one may not be able to directly twirl each individual gate in the QRAM, we show that one may nevertheless exploit the symmetry between the computing and uncomputing steps of the QRAM to perform twirling. 

Our paper is structured as follows. In Section~\ref{sec: prelims}, we review of bucket-brigade QRAM and its natural noise resilience properties. In Section~\ref{sec:extended_noise_resilience}, we show that these noise resilience properties extend, beyond the original proof, to many additional cases of practical interest. In Section~\ref{sec: Pauli twirling}, we propose a Pauli twirling scheme for QRAM. Finally, we discuss our results and conclude in Section~\ref{sec: Discussion}.

\begin{figure}[hbtp]
    \centering
    \includegraphics[width=0.9\linewidth]{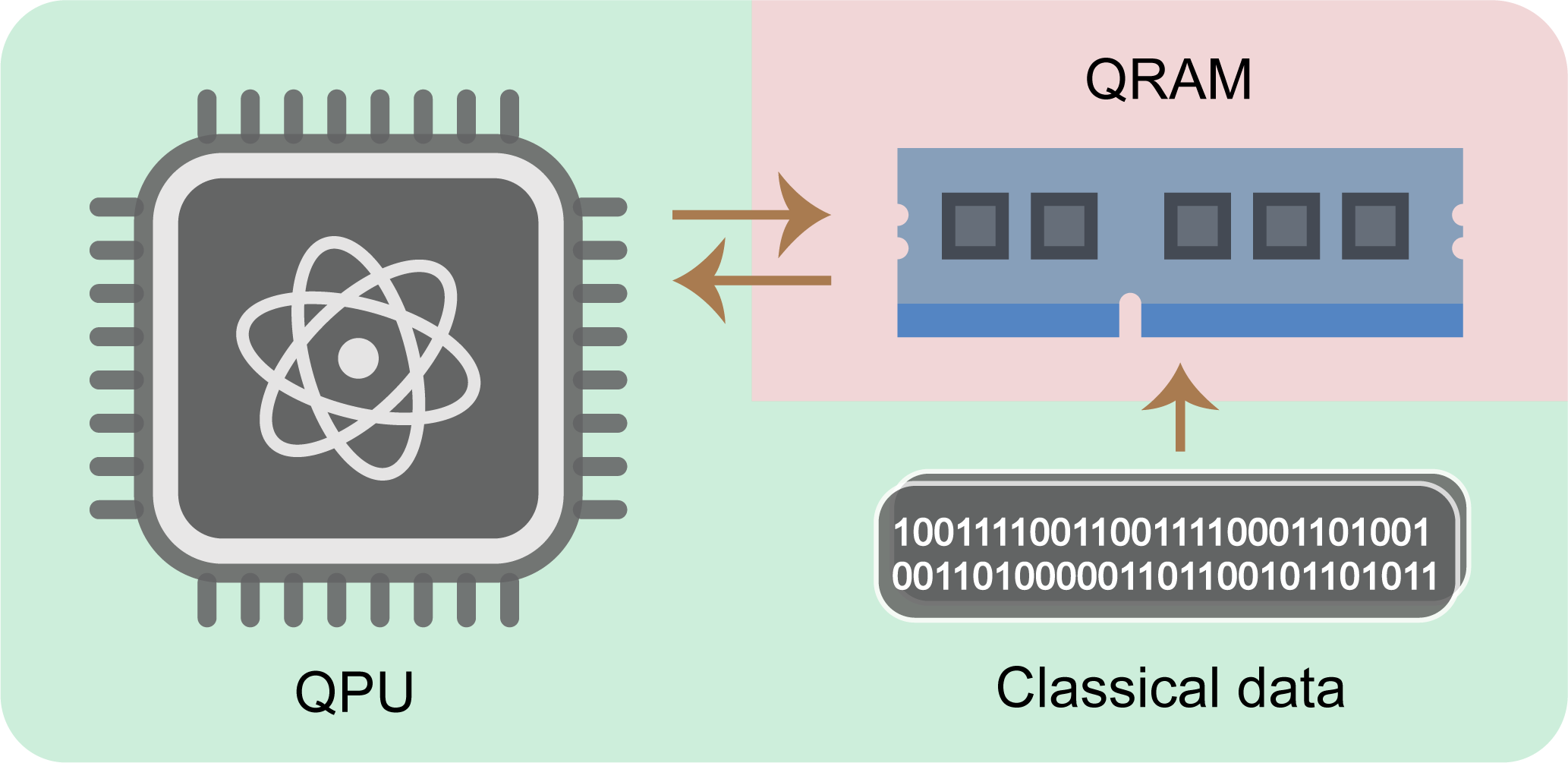}
    \caption{Conceptual overview of the role of quantum random access memory in a quantum computer architecture. The QRAM mediates interactions between the QPU and a classical database. Bidirectional arrows between the QPU and QRAM indicate that the processor could solicit several queries throughout the algorithm. Since the QRAM cannot write data to a classical memory in superposition; it can only read classical data. The green zones indicate very low noise in both the QPU and the classical memory due to error correction, whereas the red zone for QRAM indicates higher error rates.}
    \label{fig:QRAM in quantum computer architecture}
\end{figure}

\section{Preliminaries: Bucket-Brigade QRAM}\label{sec: prelims}
\subsection{Overview}\label{sec: Three-level QRAM Introduction}
Similar to classical random access memory, a quantum random access memory architecture shuttles the bus qubit to the desired mode of a database, where it can locally copy the relevant information. In the event that $\ket{\psi_{\textit{in}}}$ comprises a superposition of addresses, each query is realized coherently, such that the destination of the bus qubit is entangled with each binary address states. Classically-controlled rotations indiscriminately encode data onto several designated qubit registers, wherein the state of the bus may be stored depending on the address. Each $x_i$ data entry in the memory is an independent bit, requiring $O$ to generate complex entanglement. This entanglement is realized spatially, enabling local operations to act selectively on individual elements of the address superposition.

The bucket-brigade scheme realizes this spatial entanglement structure through a series of so-called quantum routers, $r$, each of which can be realized with a constant number of fixed-dimension qudits equipped with a routing operation $U$. If the ``control state'' of the router, $\ket{r_c}$, occupies the ground state, $\ket{W}$, the router idles; hence, we refer to $\ket{W}$ as ``passive.'' Otherwise, depending on whether the control state is $\ket{0}$ or $\ket{1}$, the router navigates an incident quantum signal in distinct spatial directions corresponding to two output modes. We refer to the mode associated with $\ket{0}$ ($\ket{1})$ as the left (right) output. Accordingly, we label the subspace spanned by $\ket{0}$ and $\ket{1}$ as ``active.''

We refer to a router capable of resolving all three of these routing states as a ``three-level router.'' A ``two-level'' router accesses only the active routing subspace, $\ket{r_c} \in \text{span}\{\ket{0}, \ket{1}\}$. Numerous physical structures for quantum routers that realize such an operation have been proposed \cite{Hann_PRL_2019_hybrid_QRAM, Hong_PRA_2012_robust_QRAM, Moiseev_JMO_2016_timebin_QRAM, Wang_arxiv_2024_transmon_QRAM, Wang_PRApp_2021_SC_QRAM_expt, Chen_PRXQ_2021_spinphoton_QRAM, Weiss_PRXQ_2024_3D_QRAM, Asaka_QST_2021_walker_QRAM}; we will consider an equivalent model where our routers are dual-register. Each router consists of two qudits, the control qudit $r_c$, and a hold qudit $r_h$. Between adjacent quantum routers connected as a depth-1 binary tree, $U^{\vec{r}}$,  
\begin{equation}
    \vec{r} := (r_c, r_h, r_0, r_1)
\end{equation}
can be implemented with two controlled-SWAP gates (Fig. \ref{fig: Routing}). In the event that $r$ is a three-level router, qubit gates should be understood as acting on only the active subspace, with trivial action on the wait state. For example, we will use $X$ to denote a qubit bit flip in two-level routers, and simultaneously in three-level routers, $X = \ketbra{W}{W} + \ketbra{0}{1} + \ketbra{1}{0}$.

\begin{figure}[hbtp]
    \centering
    \includegraphics[width=\columnwidth]{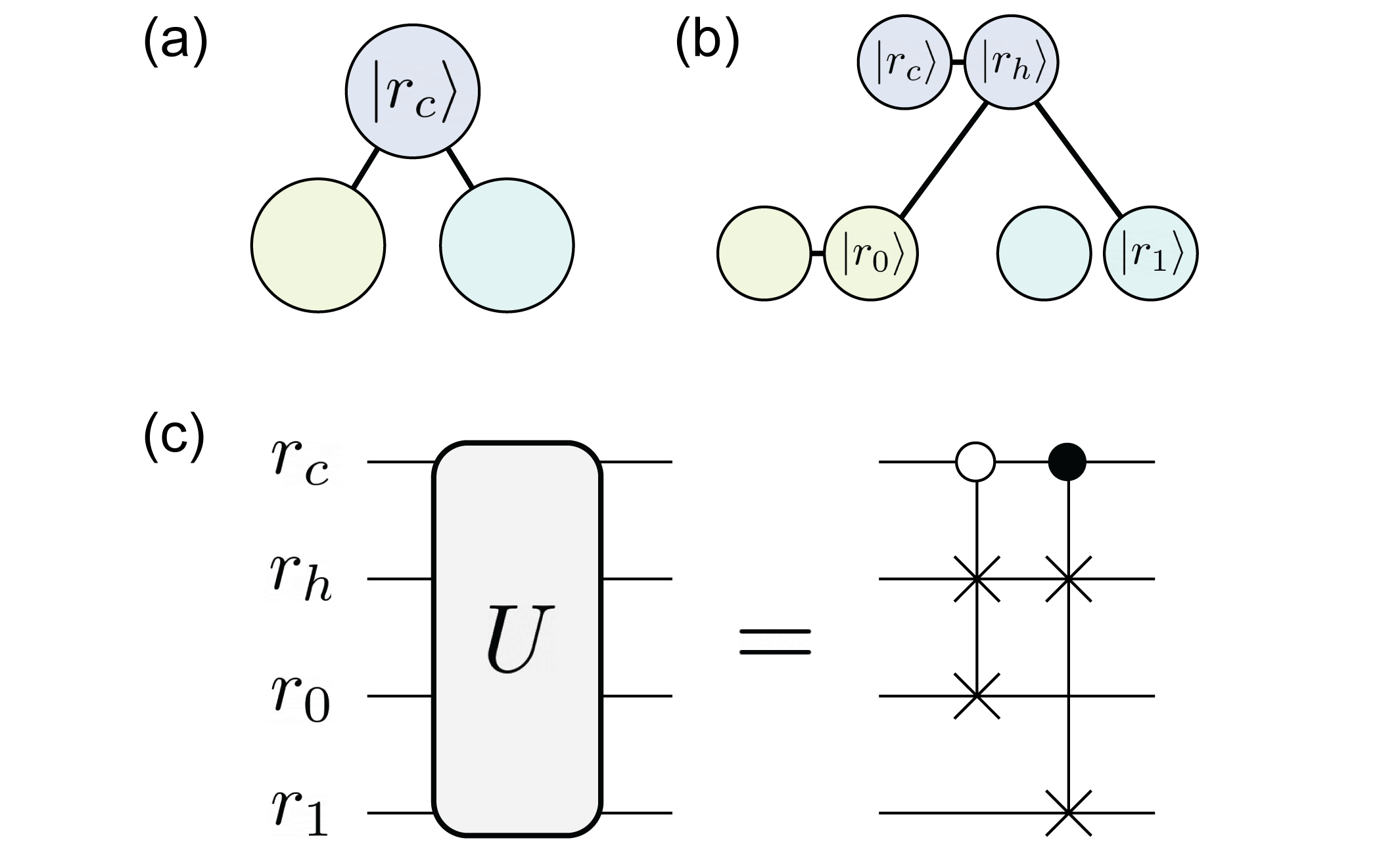}
    \caption{Microscopic depiction of a quantum router. (a) A depth-1 bucket-brigade QRAM tree with a root router $r$ that can support $n = 1$ address qubits of routing information and a single-qubit bus. Edges denote support for inter-router gates. (b) Each router has been split up into its constituent registers. Namely, we distinguish the control register, $r_c$, and the hold register $r_h$, that together form $r$. The hold registers of $r$'s children qubits are respectively labeled $r_0$ and $r_1$. In subsequent illustrations, we adopt the convention that routers depict only the state of their control register. (c) A synthesis of the routing gate on $\vec{r}$ with two controlled-SWAPs.}
    \label{fig: Routing}
\end{figure}

Stitching together a set of $N$ dual-mode quantum routers $\mc{R}$, we can construct the interaction graph of bucket-brigade architecture \cite{Giovannetti2008} that describes locally applicable two-qubit gates as a perfect binary tree $G = (\mc{R}, E)$. The query circuit of the QRAM consists of sequentially injecting address qubits into the root router of the tree, followed by routing operations. At the start of the circuit, all of the $2N$ qutrits comprising the routers in $\mathcal{R}$ are initialized to $\ket{W}$. Each $m$th address qubit makes its way to the $m$th level of the binary tree through a sequence of routing operations. When it reaches the $m$th level, a SWAP operation moves it from the hold to the control register. This way, the qubit of encoded information on each address becomes available for routing deeper into the tree. A query to any address index $i \in \mathbb{F}_2^n$ excites only the routers in a corresponding ``branch'' of the tree $\mc{R}_i$, which is the set of qubits that lie along the path from the root router to $x_i$ (Fig. \ref{fig: Bucket-brigade scheme}(d)).  This way, only $O(n)$ qutrits are excited into the active subspace. Finally, the bus qubit(s) are routed to the leaf layer of the binary tree, adjacent to which the classical memory $\vec{x}$ is stored. Each leaf router is associated with a unique entry in this memory array, $x_i$. To copy the corresponding data, local, classically-controlled $Z$ rotations are performed on the leaf routers. The reverse routing operation extracts the address and bus qubits from the tree, disentangling them completely from the tree, which is restored to its initial state under perfect operations. The explicit circuit is shown in Fig. \ref{fig: Bucket-brigade scheme}. Since gates performed on different levels of the tree may be parallelized, this query circuit takes $O(n)$ time. 

\begin{figure*}
    \centering{}\includegraphics[width=1.8\columnwidth]{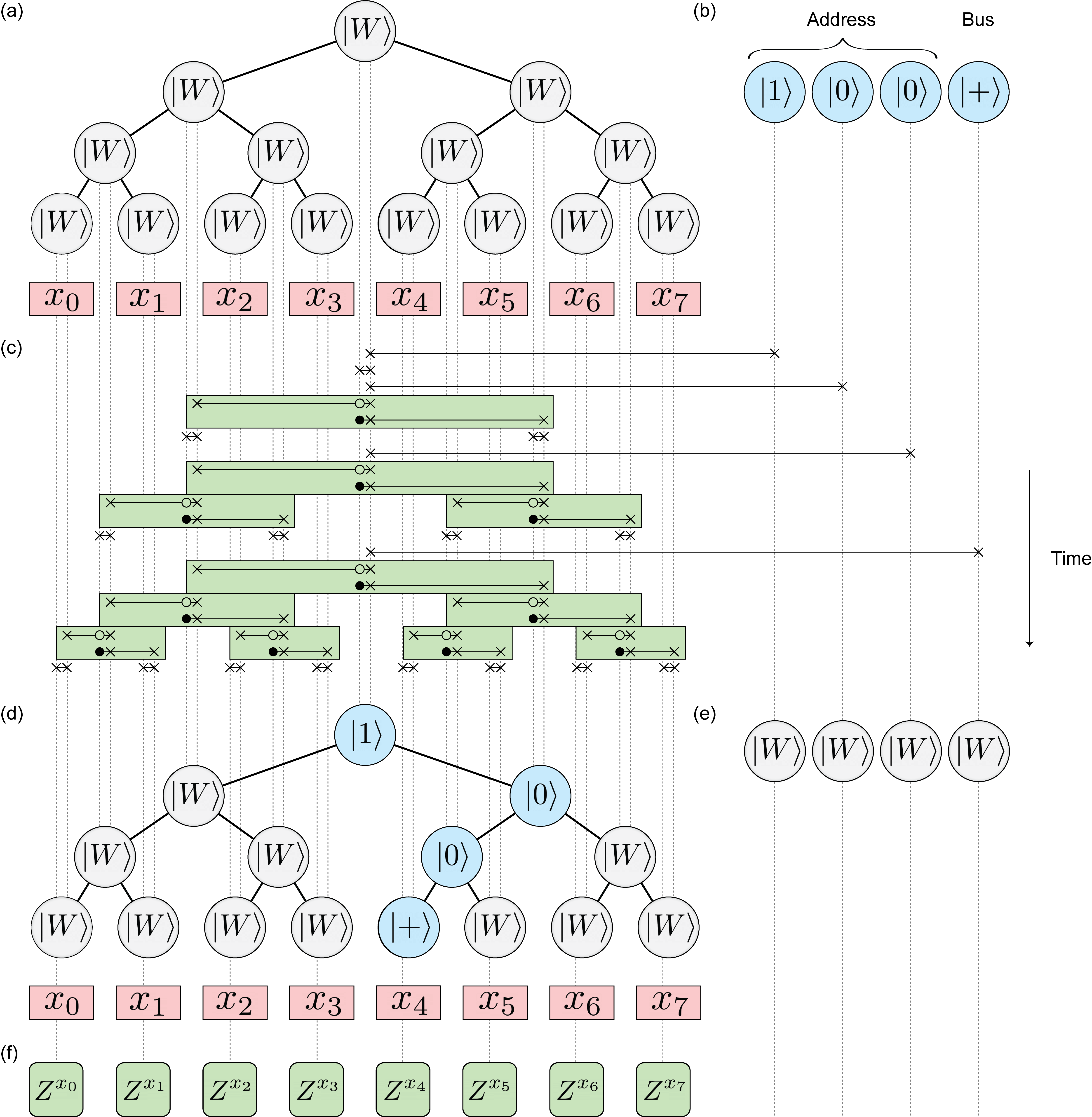} 
    \caption{Blueprint of the bucket-brigade scheme. Circuit diagrams for the query (b), (f) are depicted alongside the hardware architecture.
    (a) Query tree for $n = 3$ initialized to $\ket{W}^{\otimes 2N}$; each router displays the state of its control register. Two circuit lines (dotted gray) emerge from each router, with the left (right) line corresponding to the control (hold) qutrit register of the router. Leaf routers are paired with unique memory cells of a $\dim \vec{x} = 8$ bit memory. (b) Address and bus encoding $\ket{\psi_{\textit{in}}} = \ket{100}\ket{+}$. (c) Unparallelized circuit diagram for $V_d$. Each green block corresponds to an implementation of the routing unitary. (d) State of the tree and (e) address qubits after the completion of $V_d$. (f) Circuit implementation of $V_{\vec{x}}$. Each $Z$ gate is raised to the power of the memory cell value at $x_i$, such that a bit of classical information is transcribed into the $\{\ket{\pm}\}$ basis of the bus.
    }
    \label{fig: Bucket-brigade scheme}
\end{figure*}

Denoting the query circuit as $Q$, the ideal implementation of $O$ can be written as $O\rho_{\textit{in}}O^{\dag} = \Tr_{\mc{R}}(Q \rho_{\textit{in}} Q^{\dag})$, where $Q$ is synthesized from  sequential layers of routing operations $V_t$, $t \in \{1, 2, \dots, \tau\}$ and a $Z$-type Pauli on the $\sim N/2$ leaf routers determined by $\vec{x}$, which we will refer to as the copying gate.
\begin{equation}
    Q = V_\tau \dots V_{\frac{\tau}{2} + 1} V_{\vec{x}} V_{\frac{\tau}{2}} \dots V_{1} \label{eq: Q}
\end{equation}
The scheme used to inject and route the address and bus qubits to their destinations before the copy gate, $V_{\frac{\tau}{2}} \dots V_1$ is reversed to extract these from the tree, $V_{\tau} \dots V_{\frac{\tau}{2} + 1}$, so that routing exhibits a symmetry $V_{t} = V_{\tau - t}$. Thus, we will define the downstream routing procedure $V_d = V_{\frac{\tau}{2}} \dots V_1$ so that the extraction procedure is simply $V_d^{\dag}$, i.e., $Q = V_d^{\dag} V_{\vec{x}} V_d$.

Without prior knowledge of the address state, we must perform the routing circuit uniformly across the width of the tree. When the entire address register has been offloaded, the tree is in a superposition where each branch forms a contiguous path to the data, whereas the remainder of the tree remains passive. As the address length increases, successive levels of the QRAM tree grow exponentially in qubit count, as does the expected number of error events. This is counterbalanced by the fact that an error event in an individual branch contributes exponentially less to the query fidelity. Roughly, the expected fraction of corrupted branches at each layer remains constant, which translates to a constant fraction of successful queries, and thus, in a single instant, the infidelity growth is bounded as linear in the number of such noisy layers. Key to the robustness of the bucket-brigade model is the feature that an exponentially small fraction of the tree is ever entangled with each address destination. This low entanglement ensures that in the presence of decoherence, dephasing between queries to distinct branches is also exponentially suppressed.

Let us denote the set of address qutrits as $\mathcal{A}$ and the bus as $\mc{B}$. In what follows, we will use  $\mc{R}$ to refer interchangeably to both the $N$ quantum routers and their $2N$ underlying qutrits.
\begin{definition}
    Let $U$ denote the query unitary on the $n$ address qubit and a $2N$ ancillary qutrit input state 
    $\sigma_{\textit{in}} = \ketbra{\psi_{\textit{in}}}{\psi_{\textit{in}}} \otimes \rho^{\mc{R}}$, such that $\Tr_{\mc{R}} U \sigma_{\textit{in}} U^\dag = O \sigma_{\textit{in}} O^\dag$. Letting $\widetilde{\mc{Q}}$ denote a noisy realization of $Q$, we define the query fidelity as 
    \begin{equation}
        F_q := \braket{\psi_{\textit{out}} | \Tr_{\mc{R}} (\tilde{\sigma}_{\textit{out}}) | \psi_{\textit{out}}}. \label{eq: Query fidelity}
    \end{equation}
    where $\tilde{\sigma}_{\textit{out}} = \tilde{Q}(\sigma_{\textit{in}})$
 \label{def: Query fidelity}
\end{definition}
We assume the address and bus registers to be noiseless throughout the circuit, though we will revisit this assumption in Section \ref{sec: Discussion}. 
\begin{definition}
    For a quantum channel $\mc{E} \in \textnormal{CPTP}(\mc{H})$, we define its error rate as
    \begin{equation}
        \varepsilon := \min_{\{K_\mu\}} \max_{\psi} \left(1 - |\Re \langle K_0 \rangle_{\psi}|^2\right), \label{eq: eps}
    \end{equation}
    where $\{K_{\mu}\}$ is minimized over the set of Kraus representations of $\mc{E}$. This bounds the distance from $\mc{E}$ to the identity map as a consequence of the inequality
    \begin{equation}
        \max_{\rho} \left(1 - F(\mc{E}(\rho), \rho) \right) \leq \varepsilon,
    \end{equation}
    where $F$ denotes the state fidelity. 
    \label{def: Error rate}
\end{definition}

Def. \ref{def: Error rate} is equivalent to that of the error rate in \cite{Hann2021} for a system of only constant size---a distinction that will become important in Sec \ref{sec: coherent noise}. For convenience, we assume that the spectrum of the channel's Choi matrix is dominated by the principal Kraus operator, such that the distance between $K_0$ and the identity is much smaller than unity. Operationally, $\min_{\psi} |\text{Re}\langle K_0 \rangle_{\psi}|$ gives the smallest singular value of the Hermitian component of $K_0$, corresponding to the maximum contraction of the state space caused by the backaction of measuring the effect $K_0^\dag K_0$.

Every layer of gates $V_t$ and $V_{\vec{x}}$ is succeeded by a noise channel $\overline{\mc{E}}_t^{\mc{R}}$ or $\overline{\mc{E}}_{\vec{x}}^{\mc{R}}$. The superscript denotes that these errors act on the Hilbert space of the routing qutrits, and the overline denotes that the composite channel is supported on the entire tree. We model these errors separable on individual routers, i.e.,
\begin{equation}
    \overline{\mc{E}}_t^{\mc{R}} = \bigotimes_{r \in \mc{R}} \mc{E}_t^{r}, \label{eq: Noise separability}
\end{equation}
as in \cite{Hann2021}. One of the results of this work, as presented in Section \ref{sec: Correlated Errors} is to show that this assumption may be relaxed. Inserting noise into \eqref{eq: Q},
\begin{align}
    \tilde{\sigma}_{\textit{out}} & = \overline{\mc{E}}_{\tau}(V_\tau \dots  \overline{\mc{E}}_1 (V_1 \sigma_{\textit{in}} V_1^\dag) \dots V_{\tau}^\dag).
\end{align}
Suppose that every $\mc{E}_{t}^{r}$ is a Bernoulli channel (Definition \ref{def: Bernoulli noise}) and therefore can be written as an expectation value over a random variable $\chi_{rt}$, which is the corresponding entry of a random matrix $\chi \in \mathbb{F}_2^{N \times (\tau + 1)}$ that indexes whether each candidate error location throughout the algorithm observes the identity or pure noise. Likewise, $\chi_t$ gives a length $N$ vector distinguishing these components at time $t = 1, 2, \dots, \tau$, so $\mc{E}_{\chi_t}$ gives a composite channel acting on $\mc{R}$. By independence, the probability of any configuration $\chi$ is
\begin{equation}
\Pr(\chi) = \prod_{r, t} \Pr(\chi_{rt})
\end{equation}
such that we can write the output state as a classical distribution over density states $\tilde{\sigma}_{\textit{out}} = \E_{\chi}\left[ 
    \tilde{\mc{Q}}_{\chi}\left(\sigma_{\textit{in}}\right)
    \right]$ where
\begin{equation}
    \tilde{\mc{Q}}_{\chi}(\sigma_{\textit{in}}) = \overline{\mc{E}}_{\chi_t}(V_\tau \dots  \overline{\mc{E}}_{\chi_1} (V_1 \sigma_{\textit{in}} V_1^\dag) \dots V_{\tau}^\dag).
\end{equation}

We arrive at Definition \ref{def: Query fidelity} noting that $\mc{R}$ is just an ancillary system used to implement $O$ on the subsystem $\mc{AB}$.
\begin{definition}
    A channel $\mc{E} \in \textnormal{CPTP}(\mc{H})$ is Bernoulli if and only if it can be written as a mixture of an arbitrary noise channel $\mc{N}$ with probability $p$ and the identity map $\mc{I}$ with probability $1 - p$ for some $p \in [0, 1]$. 
    \begin{equation}
        \mc{E} = (1 - p)\mc{I} + p \mc{N}
    \end{equation}
Equivalently, we may define a random variable $\chi \sim \textnormal{Bernoulli}(p)$ such that $\mc{E} = \E_{\chi}[\mc{E}_{\chi}]$. 
    \label{def: Bernoulli noise}
\end{definition}

We now observe some properties of the QRAM that make it robust under the model of Bernoulli noise. Suppose during a query, we randomly sample some configuration $\chi$ and consider how it acts on a query with some address state $\ket{i}$, where $i \in \mathbb{F}_{2}^n$. This query to the $x_i$ memory cell succeeds if the bus qubit reaches the corresponding leaf router, copies the locally stored data, and is extracted along with the address qubits in such a way that no error occurs in the branch where this process occurs, $\mc{R}_i$ (Fig.~\ref{fig: Bucket-brigade scheme}). Because of the trivial backaction associated with $\chi_{rt} = 0$ for all $t$ and all $r \in \mc{R}_i$, at the conclusion of the circuit, the final state of the branch returns to $\ket{\mc{R}_i} = \ket{W}^{\otimes 2N}$.

If a router $r$ supports an error at some point in the circuit, then all queries to a set of branches $\{\mc{R}_k: r \in \mc{R}_k\}$ can be corrupted, the size of which is correlated with how close to the root of the tree $r$ is situated. This identifies a subspace of the address qubit Hilbert space, $\text{span}\{\ket{k}: r \in \mc{R}_k\}$ as being predisposed to encountering an error. Generalizing to more than one possible fault, for a given $\chi$, we let the subspace of the initial address Hilbert space that leads to an error-free branch as
\begin{equation}
    V_{\chi} = \text{span}\{\ket{i}: \chi_{rt} = 0 \ \forall t, r \in \mc{R}_i\}.
\end{equation}
The orthogonal subspace, $V_{\chi}^{\perp}$, denotes the subspace of $\mc{H}^{A}$ that leads to errors under $\chi$. We will also define the image of $V_{\chi} \otimes \{\ket{+}\}$ under $O$ as $\text{Im}(V_{\chi} \otimes \ket{+})$.
We might be tempted to simply estimate the expected fractional volume of $V_{\chi}$ in $\mc{H}^{\mc{A}}$ to bound the fidelity, but this requires that errors do not drive transitions between $V_{\chi}$ and $V_{\chi}^{\perp}$. Fortunately, errors for every $\chi$ map disjointly on these complementary subspaces
\begin{align}
    \tilde{\mc{Q}}_\chi: & \ V_\chi \otimes \ket{+} \ket{\mc{W}} \rightarrow \text{Im}(V_{\chi} \otimes \ket{+}) \otimes \mc{D}(\mc{H}^{\mc{R}}) \label{eq: Good subspace}\\
    \tilde{\mc{Q}}_{\chi}: & \ V_{\chi}^\perp \otimes \ket{+}\ket{\mc{W}} \rightarrow \mc{D}(\mc{V}_{\chi}^{\perp} \otimes \mc{H}^{\mc{B}} \otimes \mc{H}^{\mc{R}}) \label{eq: Bad subspace}
\end{align}
because address states that intersect errors in the tree entail necessarily orthogonal routing information from those that avoid the fault. We have let $\mc{D}(\cdot)$ denote the space of density matrices formed from convex combinations of $(\cdot)$. If an error occurs on some router $r$ located on the $l$th level of the bucket-brigade tree, the binary addresses that route to $r$ differ from those that avoid $r$ by at least one bit in the first $l - 1$ bits of routing information. The error on $r$ alone cannot alter these; an error would also have to occur further up the tree. However, we already assume that any such error leads all afflicted branches to be maximally corrupted and thus, if multiple errors occur in a branch, only the one closest to the root of the tree has any bearing.
\begin{figure}[hbtp]
    \centering
    \includegraphics[width=\columnwidth]{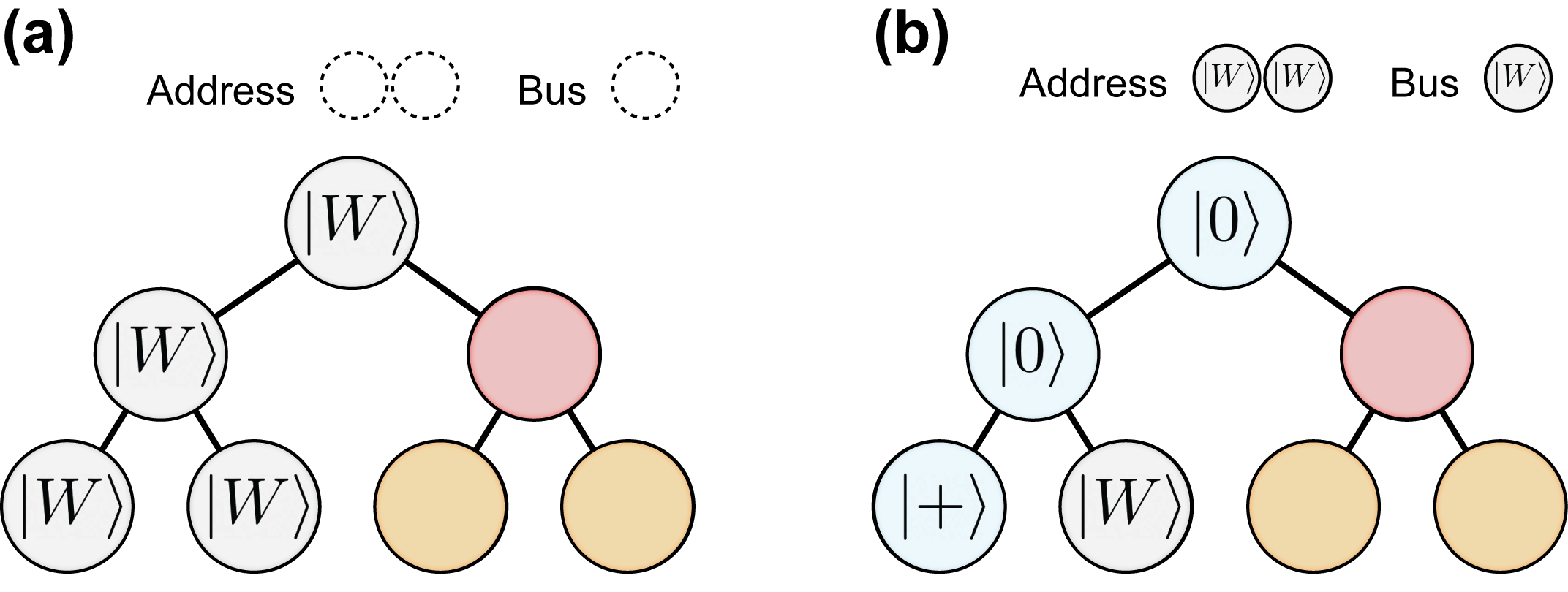}
    \caption{Constrained error propagation of a fault. The red router denotes the occurrence of an error. Even though all of its children may be error-free, we assume them all to be maximally corrupted, and thus highlight them in orange. We omit any indication of the post-error router state(s). (a) The empty-address query, $\widehat{Q}$. The gray routers remain in the wait state since no address qubits are injected into the tree. The wait state constrains the propagation of the error further up the tree. (b) A depiction of the query to $x_{0}$, where $\ket{00} \in V_{\chi}$ is the address index. Error propagation is also constrained here since the address diverges from the branches containing the fault.}
    \label{fig:Fault}
\end{figure}

One might also note that, unlike \eqref{eq: Bad subspace}, $\tilde{\mc{Q}}_{\chi}$ as written in \eqref{eq: Good subspace} preserves the separability of $\sigma_\textit{in}$ over the bipartition of $\mc{AB}$ from $\mc{R}$ in the $V_{\chi}$ subspace. This is a consequence of the constrained noise propagation afforded by the wait state in the three-level routers: the final state of the routers is identical regardless of which branch in $V_{\chi}$ is queried.

Let us define $\hat{Q}_\chi$ as a modified version of the noisy query circuit for each configuration $\chi$ where the address register is ``empty,'' which is analogous to removing the SWAP gates that span the $\mc{AB}$ and $\mc{R}$ bipartition, leaving the address and bus registers to idle noiselessly and nontrivial action locally on $\mc{R}$. Suppose that an error occurs on some $r$ on a level $l$ of the tree, which we can assume to be the closest in its branch to the root without loss of generality. This error may intersect multiple branches, but we can pick any $k$ to be any address where $r \in \mc{R}_k$. The parent router of $r$ remains in $\ket{W}$ at all times, which means that this fault can be routed at most around the subtree rooted by $r$. If we instead imagine querying to any memory cell $x_i$ where $\ket{i} \in V_{\chi}$, two cases emerge: either $i$ and $k$ differ before the $(l-1)$th bit, in which case $r$'s parent router still always remains in $\ket{W}$, or they differ on this bit, in which case $r$'s parent is excited out of the wait state into the active subspace. Nevertheless, the orthogonality of $i$ and $k$ ensures that error propagations from $r$ to its parent router remain identically constrained. For all $\ket{i} \in V_{\chi}$, 
\begin{equation}
    \hat{\mc{Q}}_{\chi}\left(\ketbra{\mc{W}}{\mc{W}}\right) = \Tr_{\mc{AB}} \left[ \tilde{\mc{Q}}_{\chi} \left(\ketbra{i, +}{i, +} \otimes \ketbra{\mc{W}}{\mc{W}} \right) \right],
\label{eq: Three-level reference} \end{equation} 
from which it follows that when $\Tr_{\mc{R}}(\sigma_{\textit{out}}^{\chi})$ is expressed in the computational basis, the submatrix where both the row and column indices are addresses $i$ and $j$ where $\ket{i/j} \in V_{\chi}$ has rank 1, up to normalization. Letting $\Pi(V)$ denote the projector onto a subspace $V$,
\begin{equation}
\Pi(V_{\chi}) \ketbra{\psi_{\textit{out}}}{\psi_{\textit{out}}} \Pi^\dag(V_{\chi}) = \Pi(V_{\chi}) \Tr_{\mc{R}}(\sigma_{\textit{out}}^{\chi}) \Pi^\dag(V_{\chi}) \label{eq: Identical on subspace}
\end{equation}
\begin{theorem}
Let $\tilde{\mc{Q}}$ denote a query with single-router Bernoulli noise acting on $\mc{R}$ and single-qutrit noise acting on $\mc{A}$ and $\mc{B}$ after each timestep, where every error channel has a Kraus representation $\{K_{\mu}\}$ such that 
\begin{equation}
    \frac{1}{\dim K_0} \Tr\left[K_0^\dag K_0\right] \geq 1 - \varepsilon,
\end{equation}
which ensures that the error rate of the channel is at most $\varepsilon$. For any input state of the form
\begin{equation}
    \sigma_{\textit{in}} = \ketbra{\psi_{\textit{in}}}{\psi_{\textit{in}}} \otimes \ketbra{\mc{W}}{\mc{W}} \label{eq: Input state for three-level QRAM},
\end{equation}
the query fidelity can be bounded as
\begin{equation}
    1 - F \leq 4 \varepsilon (\tau + 1)(n + 1) \in O(n^2).
\end{equation}
\label{thm: Three-level, mixed-unitary robustness}
\end{theorem}

\begin{proof}
The query fidelity \eqref{eq: Query fidelity} may be rewritten as
\begin{align}
    F & = \E_{\chi}\left[ 
    \braket{\psi_{\textit{out}} | \sigma_{\textit{out}}^{\chi} | \psi_{\textit{out}}}
    \right] \\
    & \geq 2\E_{\chi} \left[
    \braket{\psi_{\textit{out}} | \Pi(V_{\chi}) \Tr_{\mc{R}} \sigma_{\textit{out}}^{\chi} \Pi^\dag(V_\chi) | \psi_{\textit{out}}}
    \right]  - 1 \\
    & \geq 2 \E_{\chi}\left[|\hspace{-2pt} \braket{\psi_{\textit{out}} | \Pi(V_{\chi}) | \psi_{\textit{out}}} \hspace{-2pt} |^2\right] - 1
\end{align}
by partitioning the identity as $\mathds{1} = \Pi(V_{\chi}) + \Pi(V_{\chi}^\dag)$ and recognizing the spectrum of $\Tr_{\mc{R}} \sigma_{\textit{out}}^{\chi}$ falls within the unit interval. The last inequality results from substituting \eqref{eq: Identical on subspace}. Because
\begin{align}
\E_{\chi} \left[|\braket{\psi_{\textit{out}} | \Pi(V_{\chi}) | \psi_{\textit{out}}}|^2\right]  & = \E_{\chi} \left[ \norm{\Pi(V_{\chi}) \ket{\psi_{\textit{out}}}} \right]
\end{align}
Every branch $i$ independently is error-free with probability $1 - \varepsilon \tau n$, so
\begin{align}
\sum_{i \in \mathbb{F}_2^n} \Pr(\ket{i} \in V_{\chi}) |\alpha_i|^2 & = (1 - \varepsilon (\tau + 1)(n+1))^2 \\
& \geq 1 - 2\varepsilon (\tau+1)(n+1),
\end{align}
and thus,
\begin{equation}
    F \geq 1 - 4 \varepsilon (\tau + 1)(n + 1).
\end{equation}
\end{proof}

Generalizing beyond Bernoulli channels encounters several obstacles from our previous proof. First, if $K_0$ is not proportionate to a unitary, then every noise channel cannot be written as a classically probabilistic ensemble of no-error and noise channels. Moreover, we obtained \eqref{eq: Good subspace} by assuming that $K_0$ perfectly preserves the binary address information for any branch acted upon only by $K_0$, requiring $K_0$ to be diagonal. Finally, if $K_0 \not\propto \mathds{1}$, then \eqref{eq: Identical on subspace} falters. Yet, a statement similar to Theorem \ref{thm: Hann et al., 2021: Main} is true.

\begin{theorem}
    \textbf{\textup{Hann \textit{et al.} (2021)}}: Let $\textnormal{CPTP}_{\varepsilon}(\mc{H})$ denote the set of single-router noise channels such that for every $\mc{E} \in \textnormal{CPTP}(\mc{H})$, there exists a Kraus representation $\{K_i\}_{\forall i} \in \textnormal{Kr}(\mc{E})$ such that $K_0$ is Hermitian \footnote{We note that in restricting $K_0$ to be Hermitian, we are stating a weaker version of the result presented in \cite{Hann2021}. As we discuss in Sec \ref{sec: coherent noise}, for channels where $K_0$ is not Hermitian, the definition of Equation~(D5) in \cite{Hann2021} implicitly requires that $\varepsilon$ scales inversely with $n$, whereas it is a constant for channels with Hermitian $K_0$. When proving Theorem~\ref{thm: Coherent noise}, we will be able to restore the statement of noise resilience to its full generality in addition to extending the result to coherent noise sources.} and upper bounds the error rate of $\mc{E}$ by $\varepsilon$. For any $\sigma_{\textit{in}}$ given by \eqref{eq: Input state for three-level QRAM} where all intermediate noise channels are of the form $\mc{E}_{t}^{r} \in \textnormal{CPTP}_{\varepsilon}(\mc{H})$, the query fidelity is bounded as
    \begin{equation}
        1 - F \leq A \varepsilon (\tau + 1) (n + 1) \in O(n^2).
    \end{equation}
where $A \sim 4$ is a constant.
    \label{thm: Hann et al., 2021: Main}
\end{theorem}
\noindent The proof for this result is fairly involved, and we will not attempt to restate it here, instead referring the reader to  Appendix D of \cite{Hann2021}. In subsequent sections, we will make reference to its relevant parts only where necessary for our proofs. On a high level, we note that Theorem~\ref{thm: Hann et al., 2021: Main} relies on the fact that even in configurations where the routers in a specific branch are acted upon by a $K_0$ that is not proportional to $\mathds{1}$, the smallness of $\varepsilon$ implies that the distance $\norm{K_0 - \mathds{1}}$ is also small, as well as the fact that the parts of the query associated with error-free (up to backaction) branches remain highly coherent and are almost disentangled from the tree after the QRAM operation. We will show that similar statements can be made for the extended set of noise sources we consider.

\section{Extended noise resilience of QRAM}\label{sec:extended_noise_resilience}

Theorem ~\ref{thm: Hann et al., 2021: Main} closed a longstanding debate on the noise resilience of QRAM, which was first proposed in \cite{Giovannetti2008} and contested in \cite{Arunachalam_NJP_2015_robustness_of_QRAM}. On the other hand, in realistic quantum systems, possible sources of error may extend far beyond the local (single-router) uncorrelated noise models previously considered. However, we will now proceed to show that the natural noise resilience of QRAM also extends far beyond this error model. In particular, we demonstrate that QRAM is also resilient against arbitrary initialization errors, spatially correlated errors, as well as coherent errors. We emphasize that this high degree of noise resilience is highly unusual and is not \textit{a priori} expected for generic quantum operations, and owes largely to the particular entanglement structure of QRAM. This suggests further study of QRAM, and in particular binary routing trees, as a useful model for noise-resilient quantum operations. Furthermore, a fuller understanding of its noise resilience properties may provide a useful guide for experimental design.

\subsection{Arbitrary Initializations} \label{sec: Arbitrary initializations}
The wait states in the QRAM are useful for limiting incoherent error propagations, but can a high query fidelity maintained even without re-initializing the tree to $\ket{W}^{\otimes N}$ at the start of every query? This answer was partially answered in the affirmative in \cite{Hann2021} by introducing the idea of a two-level QRAM, where the wait state is obliviated, replacing qutrits with qubits in the active subspace. We briefly review their argument.

\begin{lemma} \textbf{\textup{Hann \textit{et al.} (2021)}}
    For single-router Bernoulli noise at all times during the query, initializing the routers to $\ket{\mc{R}} = \ket{0}^{\otimes N}$ gives the fidelity inequality 
    \begin{equation}
        1 - F \leq 2 \varepsilon (\tau + 1) (n+1)^2 \in O(n^3), \label{eq: Two-level scaling}
    \end{equation}
    which increments the degree of the polynomial. \label{lem: Hann et al. two-level}
\end{lemma}

\begin{proof}
    
We may continue to think about non-identity errors occurring probabilistically, with all possible error configurations indexed by $\chi \hspace{0pt}$. Equation \eqref{eq: Identical on subspace}, which was vital to the proof of Theorem \ref{thm: Three-level, mixed-unitary robustness}, no longer stands. Without wait states, queries to $i$ and $j$, where $\ket{i/j} \in V_{\chi}$, can both succeed in retrieving the desired memory data independently. Yet, an error located on some $r$ that is not in either branch can eliminate the coherence between these two queries in the reduced density matrix.
\begin{equation}
    \braket{i, x_i | \sigma_{\textit{out}}^{\chi} | j, x_j} \neq \alpha_i \alpha_j^*
\end{equation}
    This occurs because a router with $\ket{r_c} = \ket{0}$ exchanges the state stored in $r_h$ with $r_0$. The role of the wait state to inhibit endless fault propagation is now played in weaker fashion by $\ket{r_c} = \ket{0}$, which inhibits error propagations from $r_1$ but not from $r_0$. The final location of the fault within the tree may therefore be entangled with the branch being queried, introducing decoherence in the address-bus system. 

    In the empty-address query of $\hat{Q}\ket{0}^{\otimes 2N}$, let us assume a fault occurs on some $r$, and all of its ancestors are error-free. There exists some router $r'$ that is an ancestor of $r$ such that the support of the error remains confined to the subtree rooted by $r'$. Specifically, $r'$ is the first ancestor of $r$ that is the right output mode another router, and we denote the tree it roots by $\mc{S}_r$ (Fig. \ref{fig:Propagation dynamics}). 

    \begin{figure}[hbtp]
        \centering
        \includegraphics[width=\linewidth]{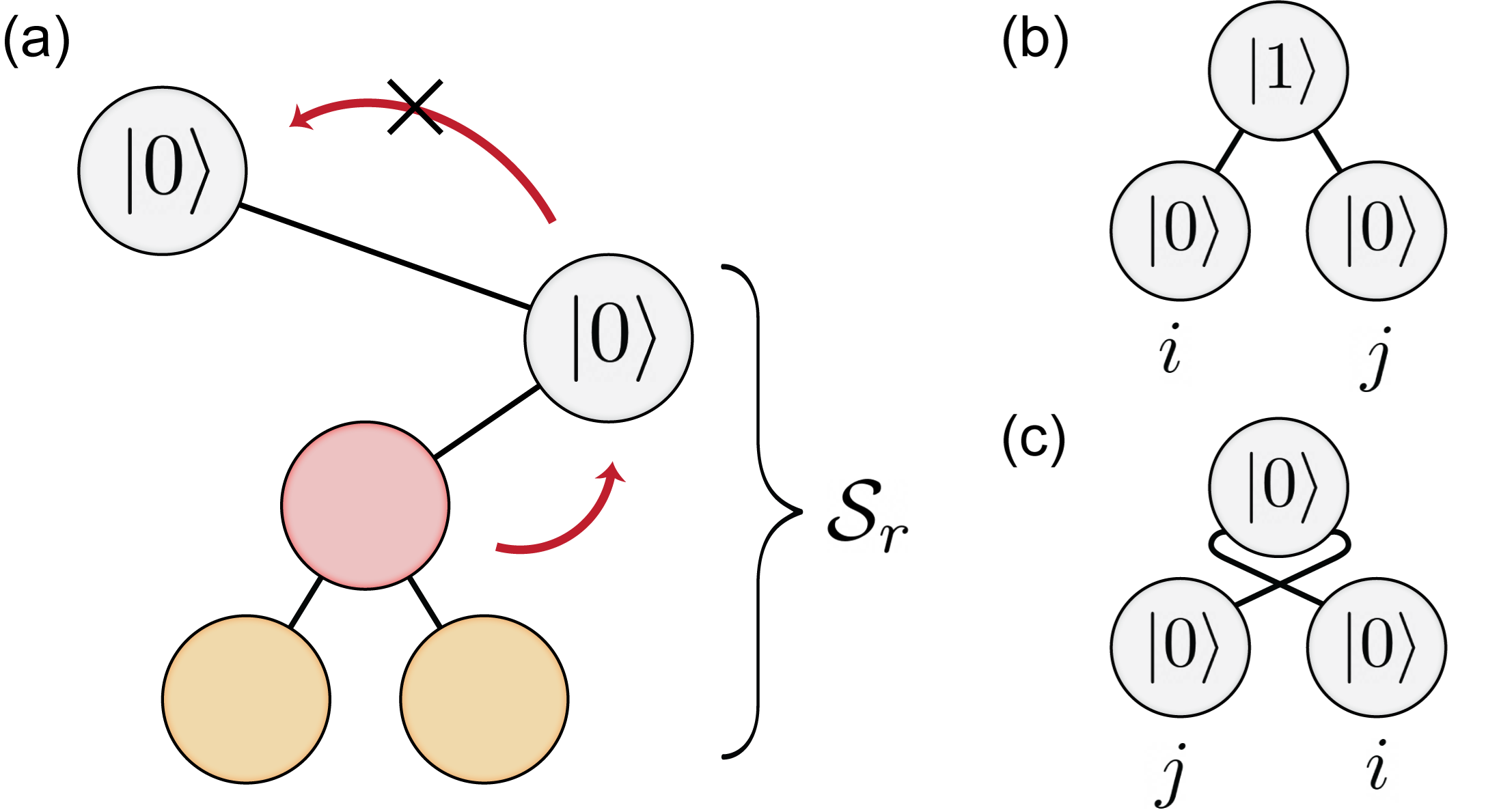}
        \caption{Error propagation dynamics for computational basis state initializations. (a) An error occurring on the red router can propagate upward one level since it is the left output mode of a $\ket{0}$ state router, but it does not propagate further. Thus, $\mc{S}_r$ is the set of routers corresponding to the subtree rooted by $r$'s parent. (b) An initialization with a router in the $\ket{1}$ state can be identified with (c) an initialization in entirely the $\ket{0}$ state up via a permutation symmetry on the branches.}
        \label{fig:Propagation dynamics}
    \end{figure}
    
    Decoherence occurs when the query to a memory node $x_i$ impedes this error propagation, for which $\mc{R}_i \cap \mc{S}_r$ is a necessary condition. Theorem \ref{thm: Three-level, mixed-unitary robustness} already factors in the effects of the noise on $r$ on the query to all $x_k$ where $r \in \mc{R}_k$, so this error can additionally only affect the coherence between $\Gamma_{\chi \hspace{0pt}}$ branches. For a configuration $\chi \hspace{0pt}$, queries where the address index is denoted by some $i$ in
    \begin{equation}
        V'_{\chi} = \text{span}\{\ket{i} \in V_{\chi}: i \in \mathbb{F}_2^n \ \text{s.t.} \ \mc{S}_r \cap \mc{R}_i = \emptyset\}
    \end{equation}
    both succeed and remain mutually coherent. Using the same approach as in Section \ref{sec: Three-level QRAM Introduction}, we obtain the correct analogues of \eqref{eq: Good subspace}, \eqref{eq: Bad subspace}, and \eqref{eq: Identical on subspace}, where we simply replace $V_{\chi}$ with $V_{\chi}'$.

In essence, the component of $\ket{\psi_{\textit{in}}}$ that lies along $V_{\chi} \setminus V'_{\chi}$ is now identified with the erroneous subspace due to the fact that errors propagate inconsistently, leading to incoherence. We will refer to $V'_{\chi}$ as mutually coherent subspace. All of the desired properties of $V_{\chi}$ used in the proof of Theorem \ref{thm: Three-level, mixed-unitary robustness} are now guaranteed to be truly only of $V_{\chi}^\perp$. We inspect that for any $\mc{R}_k$, there are $m$ routers $r$ on the $m$th level of the tree where $S_r \cap \mc{R}_k \neq 0$, so
\begin{equation}
    \Pr(\ket{i} \in V_{\chi}') \geq 1 - \frac{\varepsilon \tau}{2}(n+2)^2.
\end{equation}
for each individual query. 
After substitution, we obtain the statement of the lemma. 
\end{proof}
The advantage of initializing the QRAM to the ground state of the two-level subspace has the advantage that it does not require logical qutrits to comprise routers. However, it still requires that all routers in the tree are equipped with measurement and reset devices, which is a limiting assumption on current hardware. More conveniently, a query could be performed with the QRAM tree initialized to any arbitrary, unknown state $\rho^{\mc{R}}$. In this setting, a modification must be made to the query circuit to ensure \textit{no extra copying}: the data-copying unitary should act trivially on modes that do not contain the bus qubit \cite{Hann2021}. This was ensured in previous instances by the fact that $\ket{W}$ and $\ket{0}$ are +1 eigenstates of $Z$, hence the choice of basis for the bus. However, this guarantee is lost if the initialization is unknown. We resolve this issue by querying the circuit twice, as $Q^2 = \mathds{1}$ in a process we call ``query doubling.'' Unintended data copying operations that are residually left behind in the tree after the first query are uncomputed after the second. It would be counterproductive if this also reverted the bus, which we entangle with an ancillary qubit between the queries. Letting $\mc{B}'$ denote the new bus qubit the new circuit reads
\begin{equation}
     Q' = Q \left(\text{CX}^{\mc{B}, \mc{B}'}\right) Q \label{eq: Two-level query circuit}
\end{equation}
Since the query depth doubles, under mixed-unitary noise, so does the infidelity. 

Compared to the three-level QRAM, when the routers are in the all-zero state (or equivalently, any computational basis state), there is an additional unwanted propagation of errors due to the fact that the initial state of the tree is already associated with a particular routing instruction -- every router routes left by default. However, since this routing instruction only identifies one branch (the all-zero address) with the root of the tree, the unwanted propagation of errors remain highly constrained and identical across most branches of the tree, as shown in the above proof.

Conversely, what if the tree is initialized in an arbitrary state? For instance, consider $\ket{\mc{R}} = \ket{+}^{\otimes N}$. Error propagations are now entangled with the individual states of each router, as each router by default routes an error into a superposition of two different locations. Clearly the unwanted propagation of errors is not as highly constrained as before. Fortunately, as we now proceed to show, we may still, albeit using a different proof technique, bound the infidelity of the QRAM. 

\begin{theorem}
    Under a noisy realization of the circuit \eqref{eq: Two-level query circuit}, a two-level QRAM with the routers initialized to an arbitrary $\rho^{\mc{R}}$ satisfies the infidelity bound 
    \begin{equation}
        1 - F \leq 4 \varepsilon (\tau + 1) (n+1)^2 \label{eq: Arbitrary scaling}
    \end{equation}
    under the same noise model as Lemma \ref{lem: Hann et al. two-level}.
    \label{thm: Two-level error scaling, complete}
\end{theorem}

\begin{proof}
We will temporarily assume the router state to be pure. The subspace of address states $V_{\chi}$ is independent of the choice of basis state initializations; however, the coherent branches for a given $\chi$ now depend on $w$ as routers may propagate errors from different locations. Nonetheless, it is easy to see that the analysis for the $\ket{0}^{\otimes 2N}$ initialization immediately extends to any $\ket{\mc{R}} = \ket{w}$, $w \in \mathbb{F}_2^n$ by the branch-permutation symmetry of probabilistic fault locations shown in Fig. \ref{fig:Propagation dynamics} (b)-(c).
The $\ket{0}^{\otimes 2N}$-initialized and $\ket{w}$-initialized trees are equivalent up to local bit-flips and permutations on the classical data, both of which leave $\dim V_{\chi}'^w$ invariant over all $w$. 
    \begin{align}
        f(w, w') := \sum_\chi \Tr \bra{\psi_{\textit{out}}} \tilde{\mc{Q}}_{\chi} \left(
        \ketbra{\psi_{\textit{in}}}{\psi_{\textit{in}}} \otimes \ketbra{w}{w'} \right) \ket{\psi_{\textit{out}}} \nonumber
    \end{align}
    For a probability distribution over computational basis states, $p(\omega)$. We will let $R(p_w)$ denote this subset of the Hilbert space.
    \begin{equation}
        R(p_w) = \left\{\sum_{w} \varphi_w \sqrt{p_w} \ket{w} \in \mc{H}^{\otimes 2 N}: \varphi_w \in \mathbf{U}(1)\right\} 
    \end{equation}
    For any probability distribution $p_w$,
    \begin{align}
        \EF_{\ket{\psi} \sim R(p_w)} & = \sum_{w, w'} \int_{\mathbf{U}(1)^2} \varphi_w \varphi_{w'}^* \sqrt{p_w p_{w'}} f(w, w') d(\varphi_w, \varphi_{w'}). \nonumber
    \end{align}
    where $d(\varphi_w, \varphi_{w'})$ denotes the Haar measure over $\mathbf{U}(1)^2$. 
    We can separate $w' \neq w$ and $w' = w$ terms. In the former, $\varphi_w$ and $\varphi_{w'}$ terms vary independently over $\mathbf{U}(1)$, and therefore, for any $(w, w')$ pair, obviously
    \begin{equation}
        \int_{\mathbf{U}(1)^2} \varphi_w \varphi_{w'}^* \sqrt{p_w p_{w'}} f(w, w') d(\varphi_{w}, \varphi_{w'}) = 0.
    \end{equation}
    This leaves only 
    \begin{align}
        \EF_{\ket{\psi} \sim R(p_w)}  = \E_{w} [f(w, w)] \geq 1 - 4 \varepsilon (\tau + 1)(n+2)^2
    \end{align}
    since the lower bound on $f(w, w)$ is constant for every $w$. The factor of two increase relative to \eqref{eq: Two-level scaling} simply arises from the query doubling protocol. We will demonstrate by virtue of the fact that if this bound holds for any $\ket{\mc{R}_1} \in R(p_w)$, it must also for all initializations in $R(p_w)$. For any choice of $p_w$, any $\ket{\mc{R}_1}, \ket{\mc{R}_2} \in R(p_w)$ can be related to one another by $\ket{\mc{R}_2} = \Lambda \ket{\mc{R}_1}$, where $\Lambda$ is a diagonal phase gate in the $\{\ket{w}\}$ basis.
    \begin{equation}
        \Lambda = \sum_{w} \frac{\braket{w | \mc{R}_2}}{\braket{w | \mc{R}_1}} \ketbra{w}{w}
    \end{equation}
    We can propagate all the errors occurring at intermediate times until the end of the circuit, obtaining a nonlocal equivalent noise channel $\mathcal{X}_{\chi}$ that occurs after a successful query. 
    \begin{equation}
        \tilde{\mc{Q}}_{\chi}(\rho) = \overline{\mc{X}}_{\chi}(O^{\mc{A}} \rho O^{\dag \mc{A}})
    \end{equation}
    This way, all errors can be absorbed into a CPTP channel across the composite system occurring after an ideal query.
    \begin{equation}
        F_1 = \E_{\chi} \left[ 
        \braket{\psi_{\textit{out}} | \Tr_{\mc{R}}\left[
        \overline{\mc{X}}_{\chi}\left(\ket{\psi_{\textit{out}}}\ket{\mc{R}_1}(h.c.) \right)
        \right]}
        \right]
    \end{equation}
    By the invariance of the above partial trace under local unitaries, the equivalent equation for $F_2$ entails the replacement $\overline{\mc{X}}_{\chi}(\cdot) \rightarrow \Lambda \overline{\mc{X}}_{\chi}(\Lambda (\cdot) \Lambda^\dag) \Lambda^\dag$. We are interested in whether this affects the address-router entanglement in the $\mc{D}(\text{Im}(V_{\chi} \otimes \ket{+}) \otimes \mc{H}^{\mc{R}})$ subspace. In this subspace, $\overline{\mc{X}}_{\chi}$ has a Kraus representation $\overline{K}_{\chi, \mu}$ where
    \begin{widetext}
    \begin{equation}
        \Pi(\text{Im}(V_\chi \otimes \ket{+})) \overline{K}_{\chi, \mu}\Pi(\text{Im}(V_\chi \otimes \ket{+})) = \sum_{i \in \Gamma_{\chi \hspace{0pt}}} \sum_{w \in \mathbb{F}_2^n}\ketbra{i}{i}^{\mc{A}} \otimes \begin{cases} \overline{J}_{\chi, \mu}^{\mc{R}} \ketbra{w}{w}^{\mc{R}} & \ket{i} \in V_{\chi}'^w \\
        \overline{L}_{\chi, \mu}^{\mc{R}}(i) \ketbra{w}{w}^{\mc{R}} & \ket{i} \notin V_{\chi}'^w \\
        \end{cases}
        \label{eq: Propagated error in subspace}
    \end{equation}
    \end{widetext}
based on our analysis of error propagations. Importantly, we have made no further assumptions about the Kraus operators $\overline{J}_{\chi, \mu}$ other than their $i$-independence (for a given $w$). Any choice guarantees that for a fixed $w$ and $\ket{\phi} \in V_{\chi}'^{w}$ 
\begin{equation}
    \overline{K}_{\chi, \mu}^{\mc{AR}}\ket{\phi}\ket{w} = \ket{\phi}\overline{J}_{\chi, \mu}^{\mc{R}}\ket{w} \label{eq: Common unitary equivalence of final state}
\end{equation}
which is the key observation that guarantees the relative coherence of most error-free branches.
 Repeating the calculation with $\Lambda \overline{K}_{\chi, \mu} \Lambda$, the diagonality of $\Lambda$ in the $\{\ket{w}\}$ basis yields the same formula but with the replacement $J_{\chi, \mu} \rightarrow \Lambda J_{\chi, \mu} \Lambda^\dag$ and likewise with $L_{\chi, \mu}(i)$. Nonetheless, this still fits the form of \eqref{eq: Common unitary equivalence of final state}. This implies that the $\ket{\mc{R}_2}$ initialization is subject to the same lower bound as $\ket{\mc{R}_1}$.

 To review, we have shown that for any probability distribution $p_w$, a randomly sampled state in $R(p_w)$ obeys \eqref{eq: Arbitrary scaling}. We subsequently demonstrated that all states in $R(p_w)$ obey the same inequality as a state sampled uniformly at random from the distribution. Together, these include all parameterizations of an arbitrary pure state tree initialization $\ket{\mc{R}}$. Since the fidelity is linear in $\rho^{\mc{R}}$, it trivially follows that a classical ensemble of pure states will too, and thus Theorem \ref{thm: Two-level error scaling, complete} is proven for any arbitrary density matrix $\rho^{\mc{R}}$ that is initially unentangled with the address and bus qubits.
    \end{proof}

Even though we have assumed in Theorem \ref{thm: Two-level error scaling, complete} that the routers are initialized within the qubit space, an arbitrary three-level initialization would fare at least as well since wait states only help to constrain error propagations. Quantitatively, the explicit form of \eqref{eq: Propagated error in subspace} allows for us to generalize to sampling $w \in \mathbb{F}_3^{2N}$. But recall the wait state is defined only by the fact that it operates as a leakage state that is unaffected by the qubit subspace. This is, by definition, true of the entire Hilbert space excluding the qubit subspace---any of these states effectively operate as wait states. 

\begin{corollary}
    As an extension of Theorem \ref{thm: Two-level error scaling, complete}, the routing tree may be comprised of qudits of any, including infinite, dimension and initialized arbitrarily while still obeying the same fidelity bound. \label{col: No-reset QRAM}
\end{corollary}

The consequence of Corollary \ref{col: No-reset QRAM} is that in any architecture, the routing tree requires no reset or (weak) measurement protocols. In fact, this routing tree can be engineered without any measurement functionality whatsoever and still remain highly robust as a QRAM. This observation addresses a concern raised in \cite{Jaq2023} regarding the persistent buildup of errors in a QRAM forcing the need for router reset capabilities. Given the challenge of multiplexing measurement devices in scaling quantum processors, we expect that this architectural simplification to aid near-term implementations.

Our focus has resided on cases where we have Bernoulli errors in the tree. We later construct a Pauli twirling protocol in Section \ref{sec: Pauli twirling} that 
tailors all noise to this form, allowing for this analysis to be applicable in general. However, it is also worth mentioning that this analysis
could also extend to more general noise. The nontrivial component of considering $K_0 \not\propto \mathds{1}$ is that a small component of the query branch 
leaks into other branches of the tree. However, with arbitrary initializations, we examine only how distinct addresses can affect the routing of errors in the 
``bulk'' of the tree. Since the tree is already in some unconstrained state, a perturbation introduced by $K_0$ away from the branch being queried just result in some
other arbitrary tree state with minimal entanglement. We leave this extension to further work.

\subsection{Spatially correlated errors} \label{sec: Correlated Errors}
Thus far, we have exclusively modeled single-qubit noise that occurs after every qubit as a toy model for decoherence in QRAM. In practice, however, noise may be correlated across several routers, e.g., due to crosstalk or miscalibrated routing gates. 

Highly non-local noise correlations whose strength does not decay in the level difference between the afflicted routers can easily corrupt the fidelity of the tree. However, quasi-local noise that respects the distribution of noise sources in the tree structure could by hypothesized to still support the robustness of the error. Here, we explicitly show this to be the case in three-level QRAM systems, but this same analysis can be used to show that two-level schemes are persistent in their noise robustness under correlated noise.

\begin{definition}
    A quasi-local error $\mc{E}_t^{H}$ is \textnormal{CPTP} map acting on a connected subgraph $H = G[\text{supp}(\mc{E}_t^H)]$, where $\textnormal{supp}(\cdot)$ denotes support on quantum routers rather than their constituent qudits.
\end{definition}

For a finite quasi-locality degree $L$, we perform a cluster expansion on error channels forming connected subgraphs within the bucket-brigade tree. The error following every layer of gates is now modeled as
\begin{equation}
    \mc{E}_t(\rho) =  \hspace{-20pt} \prod_{\substack{G[S] \subseteq G \\ |S| <   L \\ G[S] \ \text{is connected}}} \hspace{-20pt} \mc{E}_{t, S}^{E}(\rho) \label{eq: Cluster expansion}
\end{equation}
where the order of composed channels matters since they may be noncommuting in general.  We demonstrate that the robustness of bucket-brigade QRAM persists, up to an $L$-dependent
rescaling of the error rate. The simplest case occurs when each error channel is Bernoulli. The lack of correlated backaction from the multiqubit $K_0$ operator allows us
to return to the counting argument of \cite{Hann2021}. Every local error cluster has exactly one supported router that is closest to the root of the tree. In the preceding analyses, we already assumed that
any error on a router immediately compromises all branches containing it. Here, we extend this by reducing correlated error channels effectively to local errors on the supported router highest 
in the tree. Clearly, the error rates are skewed up the tree, allowing the effective error rates at any router to be bounded by a function exclusively of $L$.

A more interesting error model is a correlated variant of the noise in Theorem \ref{thm: Hann et al., 2021: Main}, where we maintain the Hermiticity of all $K_0$. First, we remark that the noise robustness of the query persists in higher-dimensional variants of the bucket-brigade QRAM. We previously framed each router as resolving routing information to two output modes, which together synthesize a binary tree. Similarly, $D$-dimensional qudits can comprise the address registers and the tree routers, admitting at most a $(D - 1)$-ary routing tree if accommodating at least one wait state, or otherwise a $D$-ary tree. As a matter of concreteness, we will work through the former case with wait states. There are now $D$ basis states,
\begin{equation}
    \ket{W}, \ket{0}, \dots, \ket{(D - 1)}
\end{equation}
where the wait state functions exactly as it did before, and a computational basis state uniquely corresponds to each ``leg'' of the router. Error analyses on these trees are identical to the ones performed for qutrits. Decoherence from the no-error backaction is still bounded, and errors along the path to a memory cell remain ``worst-case'' events. Additionally, the existence of a wait state still suppresses the propagation of auxillary errors in the tree. These are the only elements used in the proofs of Theorem \ref{thm: Three-level, mixed-unitary robustness} and \ref{thm: Coherent noise}. Moreover, one can repeat the proofs for non-Bernoulli noise with Hermitian $K_0$ identically to the presentation in \cite{Hann2021} for Theorem \ref{thm: Hann et al., 2021: Main}, yielding the subsequent lemma.

\begin{lemma}
    In a generalized bucket-brigade model, where the QRAM architecture involves $d$-dimensional routers, extensions to all preceding Theorems \ref{thm: Three-level, mixed-unitary robustness}, \ref{thm: Hann et al., 2021: Main}, \ref{thm: Two-level error scaling, complete}, and \ref{thm: Coherent noise} persist with the implicit replacement $\varepsilon \rightarrow \varepsilon_d$, where $\varepsilon_d$ is a rescaled error rate parameter. 
    \label{lem:corr_noise}
\end{lemma}

Any $D$-dimensional router may itself be comprised of qutrit routers under edge contractions in the corresponding graphs. This immediately follows from the fact that one can always construct a binary tree with $D - 1$ leaf nodes. We will only concern ourselves with choices of $D$ that are integer powers of 2, where the qutrit subtree that forms a $D$-dimensional qudit is a complete binary tree. What do the wait state $\ket{W_D}$ and the logical routing states $\ket{0'}, \dots, \ket{(D - 2)'}$ look like when mapped to qutrits? 
In the strictest sense, we may define $\ket{W_D} := \ket{W}^{\otimes D - 2}$, though any state in which the path from the input to the generalized router to any of its output modes is interrupted by at least one wait state is also sufficient. 

For routing states, all qutrits that do not need to be active to route an input state correctly are fixed to the wait state. In principle, these constraints can be relaxed, as well. All remaining computational states are designated as leakage to outside of the logical space. 

\begin{definition}
    Let $d := \log(D - 1)$. For a target router dimension $D \geq 3$, a coarse graining map $\mathfrak{g}_{d, u}: \mc{R} \rightarrow \mc{R}/\hspace{-4pt}\sim$ where $u \in [1, n \bmod d]$ maps a perfect binary tree of sufficiently large depth $n$, $G$ to a tree $G'$ under edge contractions where the first $u$ and last $n \bmod d - u$ levels are comprised of uncontracted degree-3 vertices. All other vertices are $D$-regular.
\end{definition}

Let $D^{\uparrow}$ denote the maximum dimension needed to grain any noise channel. Clearly, $\log D^{\uparrow}$ is at most the size of the support of any error, $L$, which occurs if a correlated error is a path in $G$ including at most one vertex per level. Every error is single-router local in at least one grained tree. For any error channel, the quotient graph with the lowest dimensional routers satisfying such a requirement is unique, and this is the one we will use to account for the effects of such a noise channel. To order $O(\varepsilon^2)$, we can commute the noise channels in \eqref{eq: Cluster expansion} and order them such that all channels that are matched to each $\mathfrak{g}_{d, u}(\mc{R})$ act in succession. 

\begin{equation}
    \overline{\mc{E}}_t = \prod_{d, u}\left[ \bigotimes_{r \in \mathfrak{g}_{d, u}(\mc{R})} \mc{E}_{d, u}^{r}\right]
\end{equation}

Each tree $\mathfrak{g}_{d, u}(\mc{R})$ itself satisfies a bound of the form \eqref{eq: Scaling under coherent noise}, with a rescaled error rate parameter $\varepsilon_{d}$. Under this construction, the qutrit routers at the peripheries of grained trees with router dimensions $D > 3$ are treated noiselessly. Noise acting on these qutrit routers that is single-router local will be accounted for in the trivially grained graph $\mathfrak{g}_{1, 1}(\mc{R}) = \mc{R}$. As an example, in Fig. \ref{fig: Graining} depict the two nontrivial grainings of an $n=2$ bucket-brigade tree up to locality $L = 2$ with a two-router correlated error. 

\begin{figure*}[hbtp]
    \centering{}\includegraphics[width=2.0\columnwidth]{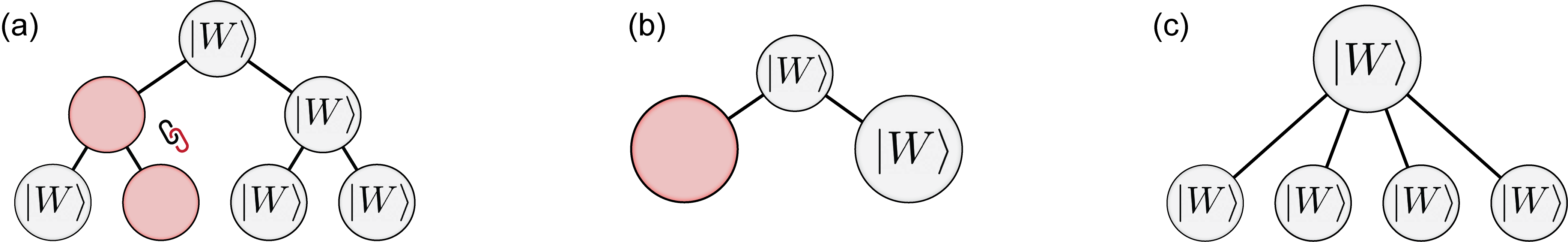}  
    \caption{The coarse grainings $\mathfrak{g}_{1, 1}(\mc{R})$ and $\mathfrak{g}_{2, \cdot}(\mc{R})$ on an $n=2$ tree. The two-router correlated error in the trivial coarse graining of (a) becomes local on a $D = 5$ router in (b), where edge contractions between the bottom two levels of the tree have been performed and root router of the tree is noiseless. Thus, this correlated error need not be considered in the graining of (c), where the top two levels of the tree are contracted, and the leaf layer is noiseless.}
    \label{fig: Graining}
\end{figure*}

We have categorized every error channel into a unique graph $(d, u)$. At any given $t$, an error analysis on every graph gives a lower bound on the fidelity $f_{d} = \min_u f_{d, u}$. The error effects from these can be taken together for a minimum overall fidelity. When all $\mc{E}_{t, S}$ can be expressed in Kraus representation with a Hermitian principal operator, we have seen fromp past sections that these infidelities can be treated additively.
\begin{align}
    1 - F & \leq \sum_{d = 1}^{L} \left[ (n+1) \bmod d \right] (1 - f_d) \\
    & \leq \sum_{d} A \varepsilon_d \tau n
\end{align}
because $u$ can assume $d$ values, and the factor of $n+1$ that arises from the ungrained tree having depth $n$ should be replaced in each grained tree $\mathfrak{g}_{d, u}(\mc{R})$ with $\lfloor \frac{n+1}{d} \rfloor \leq \frac{n+1}{d}$ levels of noisy routers. The error rate is now defined on the channel composed of all $\mc{E}_{t, S}$ that can fit into single-router errors in the tree.
\begin{equation}\label{eq:eps_d}
     \varepsilon_d \leq \max_{r \in \mathfrak{g}(\mc{R})} \hspace{-40pt} \sum_{\substack{\mc{E}_t \\ \text{supp}(\mc{E}_t) \subseteq \mathfrak{g}_{d, u}^{-1}(r) \\ 
    \text{supp}(\mc{E}_t) \not\subseteq \mathfrak{g}^{-1}_{d', u'}(r') \ \forall d' < d}} \hspace{-40pt} \varepsilon_t
\end{equation}
where $\mathfrak{g}^{-1}_{d, u}$ denotes the pre-image of $r$ in $\mc{R}$. 

In the general setting, the infidelity contribution from each graph is not simply (sub)additive, as we will discuss in the following section. In this case, one can write the composite error channel as a product of the channels that can be binned to each grained tree. Their individual infidelities can be used to bound the trace distance via the Fuch-van de Graaf inequalities \cite{Fuchs1998}, which is subadditive and obeys a triangle inequality, and can be converted back to lower bound the fidelity.

\subsection{Coherent Noise Contribution} \label{sec: coherent noise}

 Theorem \ref{thm: Hann et al., 2021: Main} is applicable to non-Pauli channels of interest, including non-unital noise such as relaxation. However, by virtue of that facts that $K_0$ is Hermitian and is close to $\mathds{1}$, it implies that $\mc{E}$ is presumed to be strongly stochastic noise. For coherent noise, we obtain a bound with larger polynomial degree in $n$.

Extending Theorem \ref{thm: Hann et al., 2021: Main} to include arbitrary channels $\mc{E}$ introduces the subtlety that noise with high \textit{unitarity}, i.e. that is dominated by interference effects \cite{Wallman2015}, can result in a fidelity that decays quadratically with iterated applications of the channel. This is true in the error estimations of generic quantum circuits \cite{Iverson2020}. However, the resilience of QRAM follows from a specialized analysis, so one might question whether this is still applicable. Indeed, we find that it is. For instance, an $n$-qubit GHZ address state,
\begin{equation}
    \ket{\psi_{\textit{in}}} = \frac{1}{\sqrt{2}} \left(\ket{0}^{\otimes n} + \ket{1}^{\otimes n} \right) \ket{+},
\end{equation}
unitary phase rotations of the form $e^{i \kappa Z}$, where $\kappa \ll 1$ on the entire tree accumulate coherently such that the fidelity decays at $O(\tau^2 n^2) = O(n^4)$. Similar coherent errors that affect the wait state can also lead to quadratically faster deviation from unity in the partial trace over $\mc{R}$ from unity.

\begin{theorem}
    The generalization of Theorem \ref{thm: Hann et al., 2021: Main} to arbitrary single-qubit noise channels, where $K_0$ need not be Hermitian, satisfies
    \begin{equation}
        1 - F \leq A \varepsilon (\tau + 1)^2 (n + 1)^2 \in O(n^4). \label{eq: Scaling under coherent noise}
    \end{equation}
\label{thm: Coherent noise}
\end{theorem}

We remark on a critical component of the proof of Theorem \ref{thm: Hann et al., 2021: Main} in \cite{Hann2021} in relation to coherent noise. In the absence of coherent noise, one can write the action of some $K_0$ as a near-identity Kraus operator on a generic state, 
    \begin{equation}
        K_0 \ket{\psi} = a(\psi) \ket{\psi} + b(\psi) \ket{\psi^\perp}
    \end{equation}
    where $\ket{\psi}^{\perp}$ lies in the orthogonal subspace to $\ket{\psi}$ in $\mc{H}$. Since $K_0$ is an observable, all $a_{\psi} \in \mathbb{R}$. Trivially, 
    \begin{equation}
    \prod_{t = \tau}^{1} \prod_{r \in \mc{R}_i} a_{t, r}(\psi) \geq (1  -  \varepsilon)^{\frac{(\tau + 1) (n+1)}{2}}.
    \end{equation}
    This allows us to bound the subsequent term that becomes relevant when one performs the partial trace calculation.
    \begin{equation}
        \Braket{\left(\prod_{t, r} K_{t, 0}^{r}\right)^2}_{\phi} 
        \geq (1 - \varepsilon)^{(\tau + 1) (n+1)} \label{eq: PT Inner Product Term}
    \end{equation}
    for any $(n+1)$ qubit separable $\ket{\phi}$. Equation \eqref{eq: PT Inner Product Term} accounts for the infidelity introduced by a nontrivial backaction. Error-free branches (for a specific configuration of Kraus operators) are imperfect since $K_0$ introduces small errors on the injected address and bus states. Furthermore, if $i, j$ are queried, $K_0$ terms acting on $\mc{R}_j$ appear to the $i$ query as perturbing the final state of the tree. However, in the query to $j$, those routers' states are stored in the noiseless address and bus registers. This leads to residual address-tree entanglement, though \cite{Hann2021} showed this entanglement to be small.
 The above argument however only holds for incoherent noise. The presence of coherent noise can be distinguished for any $K_0$ in its polar decomposition, $K_0 = V_0 P_0$, where $V_0$ is unitary and $P_0 = (K_0^\dag K_0)^{\frac{1}{2}}$ is Hermitian positive definite for $\varepsilon < 1$. Hermitian $K_0$ are already in polar form: $V_0 = \mathds{1}, P_0 = K_0$, and therefore, the coherent contribution is trivial. For instance, in a Bernoulli channel, we have $P_0 \propto I$. When $V_0$ is generated with a small but finite rotation angle $\kappa$, $|\kappa| \ll 1$, then
\begin{equation}
    \min_{\psi} | \Re \braket{K_0^{\otimes n}} | \leq n^2 \varepsilon + O(\kappa^4),
\end{equation}
from which it immediately follows that
\begin{equation}
           \Braket{\left(\prod_{t, r} K_{t, 0}^{r} \right)^\dag (h.c.)}_{\phi}
        \geq (1 - \varepsilon)^{(\tau + 1)^2 (n + 1)^2}.
\end{equation}
The right-hand side takes on the form of \eqref{eq: PT Inner Product Term} only with the replacements $n+1 \rightarrow (n+1)^2$ and $\tau + 1 \rightarrow (\tau + 1)^2$. With this modification, the rest of the proof of Theorem~\ref{thm: Hann et al., 2021: Main} follows, leading to the modified bound of Equation~(\ref{eq: Scaling under coherent noise}).

The general technique to mitigate unitary noise, which in this case increases the infidelity scaling from quadratic in $n$ to quartic, is to perform twirling. But in the QRAM setting, caveats may emerge when applying twirling and require special handling, to which we will return in Section \ref{sec: Pauli twirling}. The system scale at which this coherent noise scaling emerges dominant over the incoherent contribution from $P_0 \preceq \mathds{1}$ depends in part in on the relative magnitudes of the $V_0$ and $P_0$ components. If the coherent contribution is small or unitary effects naturally cancel, approximately quadratic scaling may still be observed.

\section{Pauli Twirling} \label{sec: Pauli twirling}
\subsection{Randomized compiling}
A common error mitigation technique is to randomize coherent errors over many iterations of a circuit, known as randomized compiling. We preface each logical gate in the circuit, $V$, with a random twirling operator, $T$, sampled from a group of unitaries. Suppose that only the implementation of $V$ is noisy as $\tilde{\mc{V}}(\rho) = \mc{E}(V \rho V^\dag)$. One then applies the correction operator $M$ that nullifies the logical effects of the twirling operator. In expectation over the distribution of all twirling gate configurations, the channel $\mc{E}$, regardless of the structure of $V$, can be randomized into incoherent noise under a sufficient choice of twirling group.

\begin{definition}
    The $V$-twirl over a group of unitaries $\mathbf{T}$ is given by a map from any CPTP map $\mc{E}(V \rho V^\dag)$ to a set of unitarily equivalent channels 
    \begin{equation}
        \{M \mc{E} (V T \rho T^\dag V^\dag) M^\dag: T \in \mathbf{T}\}
    \end{equation}
    where $\mathbf{M} \ni M = V T^\dag V^\dag$ such that for uniformly sampled $T \in \mathbf{T}$, the expected channel
        \begin{align}
        \mathscr{T}_{\mathbf{T}}^V[\mc{E}](\rho) & = \E_{T \sim \mathbf{T}}
        \left[
        M\mc{E}(VT\rho T^\dag V^\dag)M^\dag
        \right] \\
        & = \E_{T \sim \mathbf{T}} \left[
        T^\dag \mc{E}(T V\rho V^\dag T^\dag)T
        \right]  \label{eq: Twirling just noise}
    \end{align}
    projects $\mc{E}$ onto the commutant of $\mathbf{T}$ \cite{Wallman2015, Wal2016, Cai_2019}.
    \label{def: Twirling}
\end{definition}

Pauli twirling on qubits results from choosing $\mathbf{T} = \mathbf{P}_n$, the $n$ qubit Pauli group, which forms the smallest unitary 1-design in any qubit dimension. In expectation, it tailors arbitrary noise into a stochastic Pauli channel. With randomized compiling, every error channel is individually tailored in this way, so that the infidelity scales linearly in the number of channels. If we instead twirl $n$ consecutive errors in blocks of $m$, the final infidelity scales linearly in the $n/m$ independently randomized blocks, but the error rate of each block still scales quadratically in $m$. In randomized compiling, every gate, and thus each channel, is twirled as an independent subunit, but in the discussion that follows, we will consider twirling composite blocks of channels, where the size of each block constrains its effectiveness. 

We have assumed that $U$ and $M$ are noiseless in Definition \ref{def: Twirling}. In practice, this can be significantly relaxed into the condition that these gates are ``easy'' to implement, a categorization that is experimentally motivated. The easy gates should be implementable with higher fidelity than their ``hard'' counterparts, and the errors that occur during easy gate rotations are approximately independent of which easy gate is implemented, though this can be relaxed \cite{Wallman2015}. Although we have the freedom to choose an easy twirling group outright, ensuring that $\mathbf{M}$ comprises only easy gates depends also on $V$. If implementing a Clifford circuit, then $\mathbf{M} = \mathbf{P}_n$ is trivial. However, for the query circuit, we will necessarily need to append a non-Pauli gate to the easy gateset. In the canonical Clifford+$T$ basis, this is the $\pi/4$ phase gate $S = \text{diag}(1, i)$ since $[Z, T] = 0$, and $TXT^\dag = TYT^\dag = S$ \cite{Wallman2015}.
The remaining Clifford+T generators, the Hadamard gate, $H$, the $T$ gate, and a two-qubit maximally entangling gate like CZ, are classified as hard.

Applying this typical twirling protocol to the query circuit accompanies a few caveats. Firstly, the model of synthesizing $U$ from Clifford+T native gates, which we refer to as a ``digital'' QRAM implementation, demands that the native gates of a QRAM support universal quantum computation. 
\begin{equation}
    \adjustbox{valign=m}{\includegraphics[width=0.4\textwidth]{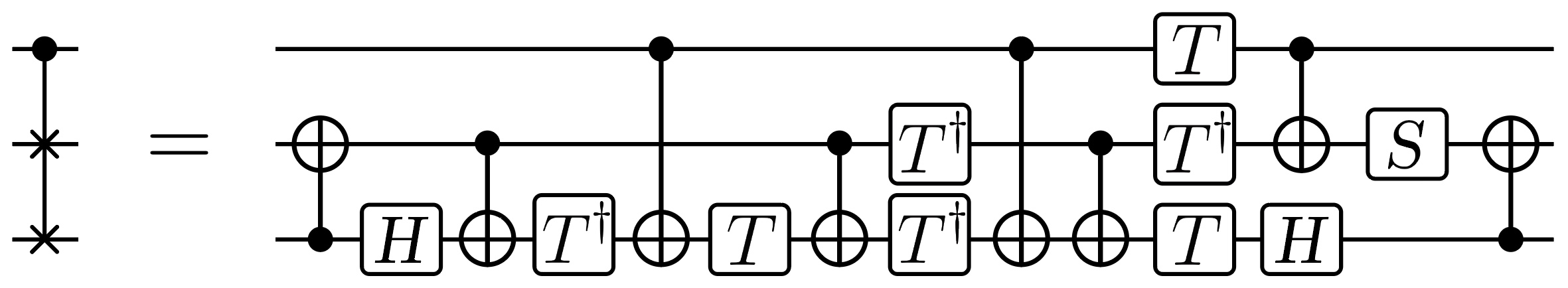}}
\end{equation}
This is true of any architecture that supports Pauli gates and synthesizes $U$ with a Hadamard-like rotation, making it equivalent to an exponentially large, universal quantum processor. In principle, it may be more natural to implement quantum routing in an ``analog'' fashion, for which the gate depth may be considerably shorter \cite{Kim2023_2, Evered2023}. In \cite{Weiss_PRXQ_2024_3D_QRAM}, two such analog implementations of $R$ are proposed on dual-rail superconducting cavities. In the first, a cavity-controlled SWAP is built directly used from beamspliters and $ZZ$ rotations. In the second, simultaneous coherent routing is realized with directional photo emission, where the photon path is entangled with the state of the resonator, such that $U$ itself is the gate primitive.  These yield higher-fidelty syntheses of $U$ than the Clifford+T basis, but these elementary gates are more challenging to twirl since $\mathbf{M}$ is not a subgroup of the easy gates. Pauli twirling a single CSWAP gate requires corrections with multi-qubit gates, which are hard operations in most architectures. Twirling $U$ as a single gate is even more complex.

In the general setting, this issue is tackled in \cite{Santos2024} with a pseudotwirling protocol that resembles dynamical decoupling \cite{Lidar2014}, but even there, systematic, e.g., miscalibration, errors persist \cite{Santos2024, Pandit2024}. To address these issues, we discuss a QRAM twirling protocol that uses only Pauli gates and SWAP to twirl a QRAM of any size, under the general gate primitive of the routing unitary. In two-level QRAM, this exactly twirls arbitrary noise to stochastic Pauli, and in the three-level protocol, this remains true at first order. Moreover, this protocol only involves two additional layers of gates relative to the two-level query circuit \eqref{eq: Two-level query circuit} for any size $n$.

\subsection{Two-level delayed twirling}

Our strategy to circumvent this issue will be to exploit the circuit's time reversal symmetry. Rather than twirl around an individual routing gate $U$ executed at time $t_d < \frac{\tau}{2}$ during the downstream routing process, we might imagine delaying the correction until we have completed its conjugate upstream routing unitary at $t_u = \tau - t_d$ supported on the same qubits. In this way, we perform a $V_{\tau - t} \dots V_t$-twirl. Imagine that we set $V_{\vec{x}} = \mathds{1}$. Then, $V_{\tau - t} \dots V_t = \mathds{1}$, and so the correction would be trivially easy since it is the same as the twirling gate. Reintroducing the data copying unitary, this twirling protocol may lead to logical failures. Namely, bit flip twirling operators on the address may route the bus qubit to the wrong mode of the memory. Only non-identity twirling operators with $T = Z$ do not interfere with the query in this way, since they commute with $Z$-basis controlled routing gates and will eventually be corrected. 

Suppose a twirling operator were $P^{r_c} = X^{r_c}$ were applied prior to some $U^{\vec{r}}$ during the downstream routing. Correcting all the twirling operators would easily entail hard gates, but as a bare minimum, we can SWAP $r_0$ and $r_1$, ensuring progress towards the correct bus mode. Note that this doesn't neutralize the twirling operators themselves; rather, it just corrects its immediate action. However, there are also residual effects on the other qubits in $\vec{r}$ that need to be tackled.

Consider a query with $\ket{\psi_{\textit{in}}} = \ket{i, +}$, and we are twirling some $U^{\vec{r}}$ where $r \in \mc{R}_i$. We might apply a twirling operator on $r_h$ before $U^{\vec{r}}$. However, we know that it will either be swapped with the $r_0$ mode if $\ket{r_c} = \ket{0}$, in which case the state stored in $r_1$ remains in place, or vice versa. In each case, note that the state stored wherever $\ket{r_h}$ does not end up is just ancillary. If $\ket{r_c} = \ket{0}$, the state in $r_0$ prior to $U_{\vec{r}}$ is swapped with $r_h$. Until at least time $t_u$, it will remain confined to the hold registers of quantum routers in $\mc{R}_i$ or the address/bus registers. The information stored there has no impact on the query. 
This observation is the consequence of the following fact: \textit{for any initialization $\ket{i, +}\ket{\mc{R}}$ and any $U^{\vec{r}}$ performed during the query circuit $Q$ such that $r \in \mc{R}_i$, the state in $r_c$ can be retraced to having been stored in an address or bus qubit register at the start of the query}. This guarantees that to ensure the bus qubit reaches its destination, we need only correct the effects of bit-flips on the states stored in the $r_c$ and $r_h$ prior to $U^{\vec{r}}$. 

In general, suppose that $r$ is in some level $\ell$ of the tree. We know that $r_h$ is being routed down the tree en route to being absorbed into a router at some $\ell' > \ell$, though we do not know at which of the $2^{\ell - \ell'}$ possible locations. Nonetheless, we can simply record if an $X$-type Pauli was applied on $r_h$ at level $\ell$, and later, when it occupies the control register somewhere on level $\ell'$ within the tree rooted by $r$, we will simply account for this twirling operator when we perform the SWAP correction. As discussed before, the correction operator that completes the twirl will be applied during the conjugate upstream routing process. At the moment, we are only describing logically repairing the twirled state.

There are several nuances relevant for twirling to be successful, which we outline. 
\begin{enumerate}
    \item \textit{Error Propagation}: One may be concerned that twirling mid-query around $V_{\tau - t} \cdots V_t$ where $t > 1$ may undermine the noise robustness of the arbitrarily-initialized two-level QRAM proven in Theorem \ref{thm: Hann et al., 2021: Main}. We expect this not to be the case. Although having applied a twirling operator may alter which routers are candidates to propagate errors, the twirling operator is a locally applied unitary that locally affects the states stored in $r_1$, whereas the propagation-induced incoherence arises from the fact that the state already being stored in $r_1$ is entangled with the branch being queried.
    \item \textit{No Extra Copying}: Returning to the second of our aforementioned criteria, since delayed twirling does not restore the tree to its initial state prior to $V_x$, we need to employ the query doubling protocol used for arbitrary initializations (Section \ref{sec: Arbitrary initializations}). However, it is also imperative that our twirling scheme does not violate the effectiveness of the compute-uncompute technique used to ensure no undesired copying of the memory data onto the tree occurs, which may lead to address-tree entanglement. If we twirl each query independently, special precautions must be taken to ensure that the non-bus modes that copied memory information all return to the correct memory cells in the second query.
    \item \textit{Twirling Errors}: In randomized compiling protocols to compile the twirling and correction gates into flanking layers of easy gates. This way, they introduce few, if any, additional errors. However, in the query circuit, these clock cycles of easy gates consist of only bit flips on the relevant control qubits. On the target modes of a controlled-SWAP, these easy gates are compiled into the identity, so implementing them physically in hardware may introduce more errors. At finite sizes with realistic, non-adversarial noise profiles, we suggest a similar protocol that is theoretically less powerful but entails only two layers of Paulis for a QRAM of any size. 
\end{enumerate}
In order to ensure no-extra copying, we must correlate our twirling operators between the first and second queries. Otherwise, we may not properly uncompute the data copying. Even if the tree is initialized in a state like two-level state like $\ket{0}^{\otimes 2N}$ that usually would not require querying the tree twice, the application of unresolved $X$-type twirling operators with delayed twirling can impart phases entangled with the address state. Correlating the twirling operators in this way contrasts with the randomized compiling approach, where the twirl for each gate is sampled IID uniformly over the group, and so the proven effectiveness of randomized compiling would be inapplicable in this setting. To avoid  introducing higher-order polynomials of the twirling operators, we note that $Q' = Q \text{CX}^{\mc{B}, \mc{B}'} Q$ exhibits a time reversal symmetry across $\text{CX}^{\mc{B}, \mc{B}'}$ just as $Q = V_d^\dag V_{\vec{x}} V_d$. 

\bigskip

\textbf{Algorithm 1:} \textit{In Situ} Delayed Twirling
\begin{enumerate}
\item We construct a randomly twirled $Q_{\text{tw}}$:
\begin{enumerate}
\item For each gate layer $V_t$ of $Q$, we sample a random Pauli $P$ on the support of $V_t$ to precede $V_t$.
\item We ``dress'' the $U$ gates from $V_t$ to $V_{\tau + 1 - t}$ by appending the necessary SWAPs as described in the main text.
\end{enumerate}
\item Obtain the twirled query unitary 
\begin{equation}
Q_{\text{tw}}' = Q_{\text{tw}}^\dag M \text{CX}^{\mc{B}, \mc{B}'} T Q_{\text{tw}},
\end{equation} where $T \in \mathbf{P}_2$ and $M = \text{CX} P^\dag \text{CX} \in \mathbf{P}_2$. 
\end{enumerate}

\bigskip

We describe the twirling protocol in Algorithm 1. Each twirling map extends for a prolonged duration, rather than just conjugating one gate as in randomized compiling. However, we can see that nested within each twirl are at most only two layers of gate noise that are not enclosed by a twirl nested further therein. Given $\mathscr{T}(\mc{E}_k \mathscr{T}(\mc{E}_{k-1} \cdots \mc{E}_2) \mc{E}_1)$, the channel $\mc{E}_k \cdots \mc{E}_2$ is randomized by the twirl that envelopes these channels, statistically independently of the outermost twirling map. By linearity, we can take the expectation over the $\mc{E}_{k-1} \cdots \mc{E}_2$ twirl, which yields a Pauli channel. The outermost twirl is only responsible for randomizing $\mc{E}_k$ and $\mc{E}_1$ on the periphery: a constant number of channels, and we know that even the outermost twirl alone yields a Pauli channel. Thus, even though twirls are prolonged to enclose a long duration of noise, recursively nested twirls randomize sub-pairings of channels at a constant frequency to disrupt coherent rotations. 

It is worth noting that the twirling operators that act on modes containing the bus during the routing process are unproblematic. In the query to a  branch $i$, the bus qubit will experience a specific sequence of Pauli twirling operators as it routed to $x_i$ that appears deterministic when observing only the action on the $\ket{i}$ address input. Since this sequence is also applied in reverse, we effectively conjugate the data copying operation $Z^{x_i}$ with a product of Pauli operators. If this Pauli product is $X$-type, for $x_i = 1$, the Pauli conjugation of $Z$ introduces a minus sign, but this sign is fixed in the second query.

As a consequence of this local twirling, coherent noise within a constant window is tailored to Pauli, which satisfies the definition of Bernoulli noise (Def. \ref{def: Bernoulli noise}). Therefore, with this twirling protocol, we may suppress the infidelity from scaling as $O(n^4)$ with coherent noise to $O(n^2)$, as given by Theorem \ref{thm: Twirling}

However, there are some practical benefits to collapsing the Pauli twirl in time, and allowing it only to randomize spatial noise correlations. In particular, at the expense of asymptotically worse fidelity scaling, we can suppress a factor of $n$ in the infidelity scaling without the need for SWAP gates. If we want to avoid performing SWAP gates, the very first time we twirl a routing gate on each $r$, we can fix a bit flip with classical corrections. Namely, we just permute the memory cells so that the following address and bus qubits approach the intended memory cell. This only suffices when no address or bus qubits have already been routed further down the branch, exploiting the same symmetry as in \ref{fig:Propagation dynamics}(b)-(c). Noting this, we can perform ``edge twirling'' \cite{Santos2024}, where Pauli operators are applied at the very beginning and the end of the circuit as a twirling technique. We depict this approach in Fig. \ref{fig:classical corrections twirl}.

\begin{theorem}\label{thm: Twirling}
For arbitrary single-router noise in line with the error model from Theorem \ref{thm: Coherent noise}, our delayed twirling approches suffice to tailor this noise to Pauli. With \textit{in-situ} corrections, we randomize coherent correlations both across different spatial error events and temporally, and thus,
\begin{equation}
    1 - F \leq 8 \varepsilon (\tau + 1) (n + 1),
\end{equation}
providing a quadratic suppression in the error scaling. In contrast, with classical corrections, we fail to randomize accumulations in time, and thus,
\begin{equation}
    1 - F \leq 8 \varepsilon (\tau + 1)^{2} (n+1).
\end{equation}
Notably, $\tau$ differs by a constant factor in each case, due to the need to physically perform SWAP gates in the former.
\end{theorem}

Even in this setting do we achieve twirled Pauli noise, as although the data copying operations we apply are adaptively dependent on the twirling operators, data copying consists of ``easy'' Pauli gates. The key difference is that here, we simultaneously twirl around coherent noise buildup across the entire algorithm. Thus, the accumulation of unitary rotations over $O(n)$ time is not suppressed, though the coherence across parallel errors in space is eliminated, leaving the infidelity to scale as $O(n^3)$.

\begin{figure}[hbtp]
    \centering
    \includegraphics[width=\columnwidth]{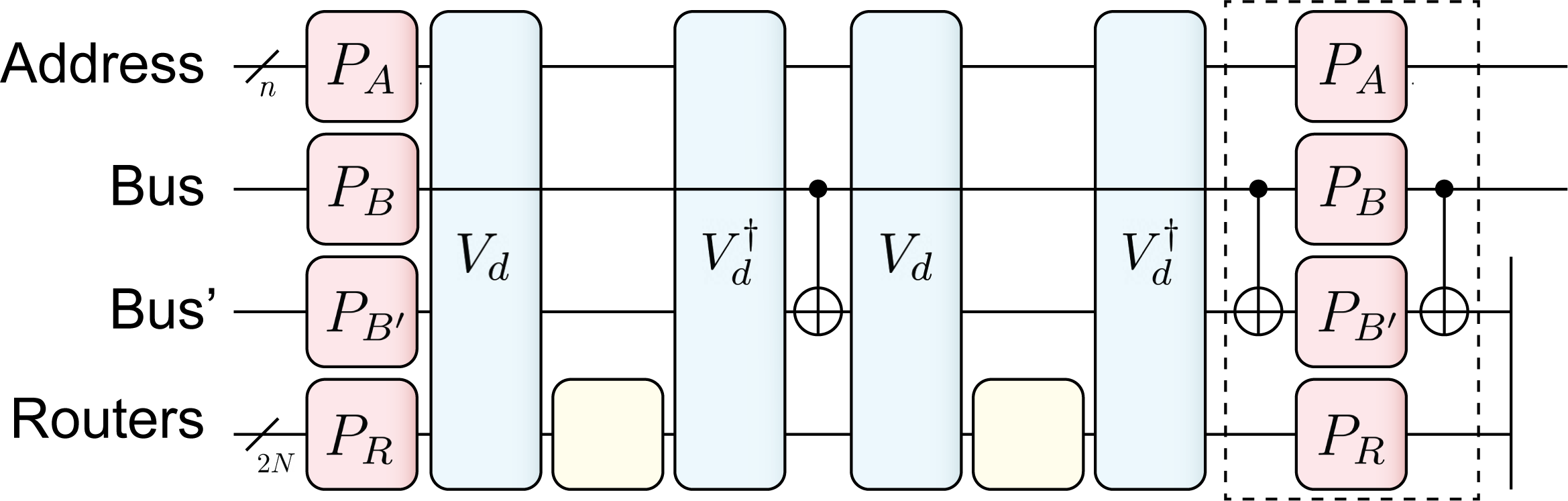}
    \caption{Delayed twirling with classical corrections. The yellow gates correspond to data copying operations, but with the memory data permuted according to Algorithm 2. The boxed product should be compiled into a single Pauli.}
    \label{fig:classical corrections twirl}
\end{figure}

\bigskip

\textbf{Algorithm 2:} Memory Reshuffle
\begin{enumerate}
\item Input: Twirling Pauli $P_m$ for the $m$th address qubit, $m = 1, \dots, n$.
\item Input: Classical memory vector, $\vec{x}$.
\item For $m$ from $1$ to $n$, if $|P_m| \in \{X, Y\}$,
\begin{enumerate}
    \item For $i \in \mathbb{F}_2^n$, exchange $x_i$ with $x_{\text{NOT}_m(i)}$.
\end{enumerate}
\item Output the resulting $\vec{x}$ as the adaptively permuted memory.
\end{enumerate}

\subsection{Challenges with three-level twirling}
The predominant issue with twirling three-level QRAM is that in addition to ensuring that the bus is routed to the correct memory cells, we must also ensure that the routers not along the address's path remains in the wait state. If we naively tried to use the twirling strategy from before, the single-qutrit twirling set would now be the qutrit Pauli group, which is $\mathbf{P}_{(3)} = \langle X_{(3)}, Z_{(3)}, e^{2 \pi i/3} \rangle$, where 
\begin{equation}
    X_{(3)} = \begin{pmatrix} 0 & 1 & 0 \\ 0 & 0 & 1 \\ 1 & 0 & 0 \end{pmatrix}, \hspace{15pt} Z_{(3)} = \begin{pmatrix}
        1 & 0 & 0 \\ 0 & e^{2 \pi i/3} & 0 \\ 0 & 0 & e^{4 \pi i / 3}
    \end{pmatrix}.
\end{equation}
Single-qutrit twirling operators $T \in \mathbf{P}_{(3)} \setminus \langle Z, e^{2 \pi i/3} \rangle$ transition routers out of the wait states. Specifically, $\langle X_{(3)} \rangle$ is isomorphic to the group of cyclic permutations on the three computational basis states. Thus, trying to edge twirl a query will take a significant fraction of the routers out of the wait state, leading to error propagations. 

Invariably, the wait state, or more generally the wait subspace, needs to be isolated from the active routing subspace. The qutrit Paulis fail as a twirling group since $\langle X_{(3)} \rangle \cong A_3$, the group of cyclic permutations on the three computational basis states, meaning no twirling operator acts as a bit flip in the active subspace and leaves $\ket{W}$ invariant. If we instead twirl over $\mathbf{P}$, the active subspace qubit Pauli group with trivial action on $\ket{W}$, coherent errors that dephase the wait state persist. Nonetheless, twirling could help suppress coherent effects. For instance, a natural embedding of logical qutrits in the routers is in the two-qubit Hilbert space, where the active subspace corresponds to the single-photon manifold. All states in $\text{span}\{\ket{00}, \ket{11}\}$ act as the wait space. A choice of two-qubit twirling group that would prevent mixing the passive and active subspaces is $\mathbf{T} = \langle \mathds{1} \otimes Z, Z \otimes \mathds{1}, X \otimes X, i \rangle$. The form of the twirled channel can be calculated from the methods of \cite{Cai_2019, Mitsuhashi2024}, where we can deduce that the resulting two-qubit Kraus operators lie in
\begin{equation}
    K \in \text{span}\{\mathbf{P}_2, e^{i \kappa (Z \otimes Z)} : \kappa \in \mathbb{R}\}
\end{equation}
because the $X \otimes \mathds{1}$-type twirling operators that randomize the two-qubit $ZZ$ parity have been excluded, leaving $Z \otimes Z$ in the center of this smaller group. Under approximately single-qubit errors, $|\kappa| \ll \varepsilon$, and although this may result in a $K_0$ with a coherent error contribution, this effect is a second-order correction, and thus, $1 - F \in O(\varepsilon n^2 + \varepsilon^2 n^4)$.

\section{Conclusion} \label{sec: Discussion}

In this work, we have shown that the noise resilience result of the QRAM extends to many experimentally relevant sources of errors, including initialization errors as well as spatially correlated and coherent noise. In each case, we have provided the modified bounds on the performance of the QRAM. In the case of coherent noise, the bound on infidelity increases quadratically (although the bound may still considered to be favourable, remaining polylogarithmic in the system size). By virtue of demonstrating the robustness of the QRAM when arbitrarily initialized, we have simplified the hardware constraints of an experimental QRAM construction, since one no longer needs to engineer reset mechanisms for an exponential fraction of the involved qudits.
\begin{table}[hbtp]
  \centering
\begin{tabular}{|c|c|c|c|c|}
    \hline
    QRAM & Initialization &  Twirled & \shortstack{Doubling} & \shortstack{$1-F$ Scaling}  \\ \hline \hline
    Two-level & Yes & No & No & $O(n^6)$ \\
    Either & No & No & Yes & $O(n^6)$ \\
    Either & Either & Yes & Yes & $O(n^3)$ \\
    Three-level & Yes & No & No & $O(n^4)$ \\
    Three-level & Yes & Yes & No & $O(n^2)$ \\
    \hline
  \end{tabular}
  \caption{Summary of error scaling results. We distinguish between cases of the three-level and two-level settings, fixed initializations, and twirling. Although not explicitly stated, a three-level QRAM without initialization is subject to the same bounds as a two-level QRAM, as discussed in Section \ref{sec: Arbitrary initializations}.}
  \label{tab:Summary}
\end{table}

Traditionally, the issues caused by coherent errors may be mitigated by randomized compiling. However, out-of-the-box approaches do not work well for QRAM. In the second part of our work, we tackle issues of this nature by exploring QRAM specific error suppression techniques. For coherent errors, we show that the original error scaling may be  restored by proposing a delayed twirling scheme that directly exploits the symmetries of the QRAM to turn coherent errors into stochastic errors. Beyond verifying that the QRAM remains robust in a wider class of situations, our proposals for reset-free QRAM and twirling mechanisms are amicable to near-term hardware limitations, where we imagine they may prove useful in facilitating a nontrivial experimental QRAM demonstration.

For future work, other than considering further methods to optimize and suppress errors in QRAM in the near-term, one should also consider how these methods can integrate with error correction. For instance, nonstabilizer codes in which CSWAPs are easily implementable may reconcile well with generality of our twirling protocol. It may also prove interesting to explore ways in which our approach to twirling could extend to other non-Clifford operations in algorithms with underlying symmetries.

\section{Acknowledgments}
We thank Connor Hann for useful discussions about uniting previous QRAM noise analyses with other types of noise, as well as with Shifan Xu for discussions about fat-tree QRAM architectures. We acknowledge support from the ARO(W911NF-23-1-0077), ARO MURI (W911NF-21-1-0325), AFOSR MURI (FA9550-19-1-0399, FA9550-21-1-0209, FA9550-23-1-0338), DARPA (HR0011-24-9-0359, HR0011-24-9-0361), NSF (OMA-1936118, ERC-1941583, OMA-2137642, OSI-2326767, CCF-2312755), NTT Research, Packard Foundation (2020-71479), and the Marshall and Arlene Bennett Family Research Program. This material is based upon work supported by the U.S. Department of Energy, Office of Science, National Quantum Information Science Research Centers and Advanced Scientific Computing Research (ASCR) program under contract number DE-AC02-06CH11357 as part of the InterQnet quantum networking project. This work was completed with resources provided by the University of Chicago’s Research Computing Center.

\bibliography{main}

\end{document}